\documentclass[aps,prl,superscriptaddress,twocolumn,10pt,longbibliography]{revtex4-2}

\usepackage{amsmath,mathtools,amsthm,amssymb}
\usepackage{times}
\usepackage{graphicx}
\usepackage{aliascnt}
\usepackage[dvipsnames]{xcolor}
\usepackage{enumitem}
\usepackage{physics}
\usepackage{array}
\usepackage{microtype}
\usepackage{caption}
\definecolor{myrefcolor}{rgb}{0.067,0.5,0.5}
\definecolor{myurlcolor}{rgb}{0.1,0,0.9}
\usepackage[
    breaklinks,
    pdftex,
    colorlinks=true,
    linkcolor=myrefcolor,
    citecolor=myrefcolor,
    urlcolor=myurlcolor
]{hyperref}
\usepackage[capitalize,noabbrev]{cleveref}

\newtheorem{theorem}{Theorem}
\crefname{theorem}{Theorem}{Theorems}
\Crefname{theorem}{Theorem}{Theorems}

\newaliascnt{lemma}{theorem}
\newtheorem{lemma}[lemma]{Lemma}
\aliascntresetthe{lemma}
\crefname{lemma}{Lemma}{Lemmas}
\Crefname{lemma}{Lemma}{Lemmas}

\newtheorem{proposition}{Proposition}
\crefname{proposition}{Proposition}{Propositions}
\Crefname{proposition}{Proposition}{Propositions}
\newtheorem{corollary}{Corollary}
\crefname{corollary}{Corollary}{Corollaries}
\Crefname{corollary}{Corollary}{Corollaries}
\newtheorem{definition}{Definition}
\crefname{definition}{Definition}{Definitions}
\Crefname{definition}{Definition}{Definitions}

\theoremstyle{remark}
\newtheorem{remark}{Remark}

\newcommand{\fu}{Dahlem Center for Complex Quantum Systems, Freie Universit\"at Berlin, 14195 Berlin, Germany}
\newcommand{\sa}{Dipartimento di Ingegneria Industriale, Università degli Studi di Salerno, Via Giovanni Paolo II, 132, 84084 Fisciano (SA), Italy}
\newcommand{\parhead}[1]{\noindent\textit{#1}}

\newcommand{\Sym}{\operatorname{Sym}}

\newcommand{\CCZ}{\mathrm{CCZ}}
\newcommand{\Tstate}{T}
\newcommand{\Mlin}{M^{\mathrm{lin}}_3}

\newcommand{\Pthree}{P_3}

\newcommand{\ketccz}{\ket{\CCZ}}
\newcommand{\kett}{\ket{\Tstate}}

\newcommand{\e}{\mathrm{e}}
\newcommand{\ii}{\mathrm{i}}

\allowdisplaybreaks

\begin{document}

\title{Universal magic state concentration}

\author{Jacopo Rizzo}
\affiliation{\fu}
\author{Lorenzo Leone}
\affiliation{\sa}
\affiliation{INFN, Sezione di Napoli, Gruppo Collegato di Salerno, Italy}

\begin{abstract}
Magic plays a dual role in quantum computation: it promotes stabilizer dynamics from efficient classical simulability to computational universality, but also challenges fault-tolerant architectures, since non-stabilizer operations are harder to protect against noise. 
Magic state distillation addresses this issue, yet existing protocols remain largely tailored to specific assumptions about the input or noise model and, more fundamentally, no general information-theoretic theory of optimal magic-state conversion currently exists.
Here we develop such a characterization for pure qubit states through universal magic state concentration: a fixed stabilizer protocol that converts a few copies of an unknown pure non-stabilizer qubit state into an exact target magic state. Motivated by the impossibility of exact $T$-state concentration, we build six- and eight-copy protocols producing an exact $\CCZ$ state, with six copies being minimal.
Their success probabilities are governed by the linearized order-three stabilizer Rényi entropy $\Mlin$. We show that this connection is structural: for up to nine input copies, $\Mlin$ fully determines the success probability of every universal Clifford-invariant stabilizer protocol.
Remarkably, this characterization persists asymptotically: our constructions achieve optimal rates among universal single-output protocols and remain optimal up to logarithmic factors among arbitrary stabilizer protocols.
As a corollary, we show that any unknown pure non-stabilizer state suffices for universal quantum computation via probabilistic $\CCZ$-state injection. Together, these results identify the stabilizer Rényi entropy as a fundamental operational quantity in magic state distillation.
\end{abstract}

\maketitle

\parhead{Introduction.---} 
Magic delineates the boundary between the efficiently classically simulable stabilizer regime and universal quantum computation, providing the essential resource behind quantum computational advantage \cite{Gottesman1998,Aaronson_2004,BravyiKitaev2005}. At the same time, it constitutes a major engineering bottleneck in fault-tolerant quantum computing. Logical Clifford operations admit particularly simple fault-tolerant implementations in many quantum error-correcting codes, often transversally, and are therefore comparatively inexpensive to protect against noise \cite{gottesman2009introductionquantumerrorcorrection}. By contrast, implementing non-Clifford gates fault tolerantly is typically substantially more demanding. More generally, no quantum error-correcting code admits a universal set of transversal logical gates \cite{Eastin_2009}. In standard fault-tolerant architectures, non-Clifford gates---most notably the $T$ gate---can consequently require substantial auxiliary-qubit and error-correction overhead \cite{Bravyi_2012,Litinski_2019}. They are therefore often implemented through magic-state injection \cite{BravyiKitaev2005}, in which the required non-stabilizer resource is supplied by a specially prepared magic state, while the remaining processing uses only stabilizer operations. The reliable production of high-quality magic states is thus a central challenge for scalable fault-tolerant quantum computation.

The key primitive in this setting is magic-state distillation \cite{BravyiKitaev2005}, whereby multiple copies of a noisy resource qubit state are processed using only stabilizer operations into fewer states of higher fidelity with a prescribed pure magic state. In seminal work, Bravyi and Kitaev \cite{BravyiKitaev2005} introduced stabilizer-code-based protocols in which five input states are distilled toward a $T$-type target and fifteen input states toward an $H$-type target, conditioned on suitable syndrome outcomes. The corresponding target states can be written as $\ket*{T_{\mathrm{BK}}}\coloneqq\cos\theta\,\ket{0}+e^{i\pi/4}\sin\theta\,\ket{1}$, with $\cos(2\theta)=1/\sqrt{3}$, and $\ket{H}\coloneqq\cos(\pi/8)\ket{0}+\sin(\pi/8)\ket{1}$. Iterating these routines drives the output toward the target magic state whenever the input fidelity towards a specific magic direction exceeds a threshold. The distilled states can then be consumed through gate injection to implement universal quantum computation. This established the template followed by conventional distillation protocols: the target state and code are fixed in advance, while selected syndrome branches suppress errors at the cost of consuming multiple inputs.

Reichardt \cite{Reichardt2005,Reichardt2009} later improved the distillation thresholds for qubits, showing in particular that every pure non-stabilizer state enables universal computation with stabilizer operations, although the required protocol may depend on the input state's Bloch-sphere direction.
For mixed states, the situation is more subtle: non-stabilizer qubit states can exhibit bound magic under finite-size distillation protocols \cite{Campbell_2010}, while exact purification of full-rank resource states is generally impossible under broad classes of free operations \cite{Fang_2020,FangLiu2022}.

Later code-based protocols progressively reduced the input-to-output overhead of distilling noisy states toward a prescribed target, under suitable noise assumptions. For iterated $[[n,k,d]]$ code constructions, with $n$ (resp. $k$) physical (resp. logical) qubits, and distance $d$, this overhead scales as $O(\log^\gamma(1/\delta))$ with $\gamma=\log(n/k)/\log d$. Multi-output and multilevel protocols progressively reduced this exponent \cite{meier2012magicstatedistillationfourqubitcode,Bravyi_2012,Jones_2013}, culminating in the first sublogarithmic construction with $\gamma<0.678$ \cite{HastingsHaah2018}. Subsequent code families made $\gamma$ arbitrarily close to zero \cite{KrishnaTillich2019,GolowichGuruswami2025}, while more recent constructions achieved constant overhead \cite{WillsHsiehYamasaki2025}.

On the other hand, a complementary line of research aims to prepare the required resource more directly, rather than first distilling simpler magic states, by combining error suppression with preparation of the desired resource \cite{Eastin2013,Jones2013,CampbellHoward2017,HaahHastings2018Codes}. Surface-code factories further reduce physical overhead through catalytic conversions and optimized code distances \cite{GidneyFowler2019,Litinski_2019}. More recently, zero-level distillation and magic-state cultivation suppress errors while preparing and growing logical magic states, before full surface-code protection \cite{Itogawa_2025,gidney2024magicstatecultivationgrowing}.

These developments, however, usually rely on prior assumptions about the input, such as a specified noise model or sufficient fidelity with a prescribed magic direction. This leaves open the broader question of whether magic can be concentrated universally, that is, without any prior knowledge of the initial resource state. More precisely, by \textit{universal} we mean a fixed stabilizer protocol whose circuit, measurements, acceptance rule, and outcome-dependent corrections are all chosen independently of the input state, requiring neither tomography nor state-dependent adaptation. Here we address this question for pure qubit states, imposing the stronger requirement of exact distillation of a fixed target magic state. By \textit{concentration}, we mean that we restrict ourselves (mostly) to pure-state inputs; this is in fact necessary for a nontrivial universal exact distillation problem, since for arbitrary non-stabilizer mixed states the worst-case success probability necessarily vanishes \cite{Fang_2020,FangLiu2022}.

More broadly, this problem lies within the realm of quantum resource theories \cite{Chitambar_2019}, which characterize transformations between free and resourceful states under a specified class of free operations. For magic, stabilizer states are free, and stabilizer operations---generated by Clifford unitaries, preparation of stabilizer states, Pauli measurements, feedforward, and discarding---form the set of free operations, with magic being the resource required for universal quantum computation \cite{Howard_2017,Veitch_2014}. Universal resource distillation has also been studied for entanglement \cite{hayashi2002universaldistortionfreeentanglementconcentration} and within more general resource theories under resource-non-generating operations \cite{lami2026universalquantumresourcedistillation}.

A useful characterization of magic is provided by the stabilizer Rényi entropies \cite{LeoneOlivieroHamma2022}. For an $n$-qubit pure state $\ket{\phi}$, let $\widehat{\mathcal{P}}_n\coloneqq \mathcal{P}_n/\{\pm I,\pm iI\}\simeq\{I,X,Y,Z\}^{\otimes n}$ denote the Pauli group modulo phase. We also define $\mathcal{C}_n\coloneqq\{U\in\mathcal{U}(2^n): \,U\mathcal{P}_nU^\dagger=\mathcal{P}_n\}$ to be the $n$-qubit Clifford group. The order-$\alpha$ stabilizer purity is $P_{\alpha}(\phi)\coloneqq 2^{ -n}\sum_{P\in\widehat{\mathcal{P}}_n}|\langle\phi|P|\phi\rangle|^{2\alpha}$, with associated entropy $M_{\alpha}(\phi)\coloneqq(1-\alpha)^{-1}\log_2P_{\alpha}(\phi)$ and linearized form $M_{\alpha}^{\mathrm{lin}}(\phi)\coloneqq1-P_{\alpha}(\phi)$. These quantities are faithful and Clifford invariant, with $P_{\alpha}$ multiplicative and $M_{\alpha}$ additive. For every integer $\alpha\geq2$, $M_{\alpha}$ is monotone under deterministic pure-state stabilizer protocols, while $M_{\alpha}^{\mathrm{lin}}$ is strongly monotone: if $\ket{\phi}\mapsto\{p_i,\ket{\phi_i}\}_i$, then $M_{\alpha}^{\mathrm{lin}}(\phi)\geq\sum_i p_iM_{\alpha}^{\mathrm{lin}}(\phi_i)$ \cite{LeoneBittel2024}. Suitable convex-roof constructions extend them to mixed states.
Stabilizer Rényi entropies also have operational meaning in property testing, quantifying distinguishability from stabilizer and Haar-random behavior \cite{Bittel_2026}. They further admit experimental estimation schemes on quantum processors and efficient numerical methods for broad classes of many-body states
\cite{OlivieroLeoneHamma2022,Tarabunga_2023,PhysRevLett.131.180401,Tarabunga_2024,Turkeshi_2025,ding2025evaluatingmanybodystabilizerrenyi}.
We also define the stabilizer group of $\phi$ as $\mathcal S(\phi)\coloneqq\{P\in\mathcal P_n:P\ket{\phi}=\ket{\phi}\}$. The stabilizer nullity is then defined as $\nu(\phi)\coloneqq n-\log_2|\mathcal S(\phi)|$. It is faithful, additive under tensor products, and cannot increase along any nonzero pure-state stabilizer branch.

In our magic-state concentration setting, we take as main target the $\CCZ$ magic state
\begin{align}
\ketccz \coloneqq \frac{1}{\sqrt{8}} \sum_{x_1,x_2,x_3\in\{0,1\}} (-1)^{x_1x_2x_3} \ket{x_1x_2x_3}.
\end{align}
The $\CCZ$ state is a natural multiqubit magic resource: its gate is locally Clifford-equivalent to Toffoli and, together with Clifford operations, enables universal quantum computation \cite{shi2002toffolicontrollednotneedlittle}. It is also closely related to the standard $T$ state, $\kett\coloneqq(\ket0+\e^{\ii\pi/4}\ket1)/\sqrt2$: while exact conversion from finitely many $\CCZ$ states to $T$ states is impossible under stabilizer operations alone \cite{Beverland_2020}, catalytic conversion $\ket{\CCZ}\to_{\mathrm{cat}} \ket{T}^{\otimes 2}$ with a $\ket{T}$ catalyst is possible \cite{GidneyFowler2019,Beverland_2020}.

In this work, we establish a sharp six-copy threshold for exact universal magic state concentration. With fewer than six copies, no fixed stabilizer protocol can convert every pure non-stabilizer qubit input into any fixed non-stabilizer target with nonzero success probability, whereas six copies suffice to universally produce one exact $\CCZ$ state. Among all universal six-copy stabilizer protocols producing a $\CCZ$ state, our protocol achieves the optimal success probability (see \cref{thm:six-copy-optimality}). We further construct a universal eight-copy protocol $\Lambda_8$ that doubles this optimal six-copy probability (see \cref{thm:eight-copy-optimality}). Explicitly,
\begin{align}
\begin{aligned}
\Pr_{\mathrm{opt}}(\psi^{\otimes 6} \to \CCZ) &= \frac{1}{3}\Mlin(\psi), \\
\Pr_{\Lambda_8}(\psi^{\otimes 8} \to \CCZ) &= \frac{2}{3}\Mlin(\psi).
\end{aligned}
\end{align}
Thus, the linearized order-three stabilizer Rényi entropy $\Mlin$ acquires a direct operational interpretation as the optimal probability of extracting one exact $\CCZ$ state at the minimal block size. Remarkably, its role extends beyond $\CCZ$ conversion: up to nine input copies, the success probability of any Clifford-invariant universal stabilizer protocol to an arbitrary fixed non-stabilizer target depends on the input state only through $\Mlin$ (see \cref{prop:rigcliff}).
For a single qubit, $\Mlin$ satisfies $0\leq\Mlin(\psi)\leq 4/9$ and attains its maximum at the eight states in the Clifford orbit of the Bravyi--Kitaev $T$-type state $T_{\mathrm{BK}}$ \cite{BravyiKitaev2005}. Catalytic $T$-conversion \cite{GidneyFowler2019} also yields the corresponding success probabilities for catalytic $T$-state concentration (\cref{cor:catalytic-t}).
As a direct consequence of our protocols, a single fixed stabilizer procedure enables exact universal quantum computation via $\CCZ$ injection from every pure non-stabilizer qubit state (\cref{cor:uqc}); only its success probability, and hence the required overhead, depends on the input state. The same protocols also apply to arbitrary six- and eight-qubit mixed states supported on the symmetric subspace, producing an exact $\CCZ$ state whenever they succeed (\cref{cor4main}).

Block repetition yields universal asymptotic $\CCZ$ distillation with exponentially vanishing failure probability (\cref{cor2:rates}). Using the additivity of the relative entropy of magic for $\CCZ$ \cite{Rubboli_2024}, we prove that these rates are optimal up to logarithmic factors even among arbitrary collective distillation protocols (\cref{prop:mainlow}). Independently of the target, we further show that any universal protocol extracting a single exact non-stabilizer state at arbitrary block size is bounded by the same $\Mlin$ scaling up to constants (\cref{prop:many-to-one}).

In the following sections, we discuss the main technical results summarized above in detail. Rigorous proofs are deferred to the appendix.

\parhead{Main results.---} Our construction has a direct geometric picture. We first define the stabilizer-orthogonal symmetric subspace
\begin{align}
\mathcal{M}_k \coloneqq \left\{ \ket*{\Phi} \in \Sym^k(\mathbb{C}^2) : \, \braket*{s^{\otimes k}}{\Phi} = 0,\, \forall s \in \mathrm{Stab_1}  \right\},
\label{eq:stabot}
\end{align}
where $\mathrm{Stab}_1$ are the six vertices of the single-qubit stabilizer octahedron, and $\Sym^k(\mathbb{C}^2)$ is symmetric subspace of $k$ qubits. More generally, $\mathrm{Stab}_n$ is the set of all $n$-qubit stabilizer states. By construction, $\mathcal{M}_k$ is Clifford-invariant under the action $C^{\otimes k}$ with $C \in \mathcal{C}_1$. We will also denote the respective subspace projectors as $\Pi_{\mathcal{M}_k},\Pi_{\Sym^k(\mathbb{C}^2)}$. These will both commute with the $k$-copy Clifford action.
Let now $\ket{\psi}=a\ket{0}+b\ket{1}$ with $a,b\in\mathbb{C}$ and $|a|^2+|b|^2=1$ be the unknown input qubit state. To each $\ket*{\Phi} \in \Sym^k(\mathbb{C}^2)$, we can associate the homogeneous degree-$k$ polynomial $p_{\Phi}(a,b) \coloneqq \bra*{\Phi}(a\ket{0}+b\ket{1})^{\otimes k} \in \mathbb{C}[a,b]_k$. Hence $\ket*{\Phi} \in \mathcal{M}_k$ precisely when $p_{\Phi}$ vanishes on the respective stabilizer rays corresponding to $\mathrm{Stab_1}$. The corresponding minimal-degree homogeneous vanishing polynomial is $f_6(a,b)\coloneqq ab(a^4-b^4)$, which is connected to the linearized order-three stabilizer Rényi entropy $\Mlin(\psi)$ by the simple relation (see \cref{app:linstab})
\begin{align}
\Mlin(\psi) = 6 |f_6(a,b)|^2.
\end{align}

Now, since any exact magic-state concentration protocol acting on $\psi^{\otimes k}$ must produce a pure non-stabilizer state on each successful branch $r$, we may fine-grain the protocol so that every such branch is represented by a single Kraus operator $K_r$. Universality then requires $K_r$, when restricted to the symmetric subspace, to vanish on all stabilizer inputs and hence to have input support contained in $\mathcal{M}_k$. The structure of $\mathcal{M}_k$, and in particular suitable basis elements, will therefore guide our choice of stabilizer projections in the explicit protocols. As we show below, these projections can be implemented by Pauli measurements whose accepted branches eliminate the input-state dependence and yield fixed non-stabilizer outputs, all Clifford-equivalent to $\CCZ$. We begin with the six-copy case.

\begin{theorem}[Six-copy threshold and optimality. Informal version of \cref{thm:sixprob,thm:opt6}]
\label{thm:six-copy-optimality}
For an arbitrary single-qubit pure state $\psi$, there exists an universal exact $\CCZ$ magic-state concentration protocol $\Lambda_6$ acting on $k=6$ copies of
$\psi$ preparing an exact $\CCZ$ state with success probability
\begin{align}
\Pr_{\Lambda_6}(\psi^{\otimes 6}\to \CCZ)=\frac{1}{3}\Mlin(\psi).
\end{align}
Furthermore, among all other $k$-copy protocols of the same class:
\begin{enumerate}[label=(\roman*),leftmargin=1.5em,itemsep=0pt,topsep=2pt]
    \item for $k=6$, the above success probability is optimal, and no accepted branch can contain two $\CCZ$ states;
    \item for every $k<6$, the optimal success probability of universal exact conversion into any fixed non-stabilizer target is zero.
\end{enumerate}
\begin{proof}[Proof sketch]
In what follows, since the input is assumed to have the form $\psi^{\otimes 6}$, we restrict our attention to the symmetric subspace $\Sym^6(\mathbb{C}^2)$; all subspaces below are understood as subspaces thereof. Consider now the subspace $\mathcal{M}_6$ defined in \cref{eq:stabot}. In the appendix, we show that the joint Pauli measurement of $X^{\otimes 6}$ and $Z^{\otimes 6}$, with respective syndromes $x_6,z_6\in\{-1,+1\}$ and associated joint eigenspaces $\mathcal{H}_6^{(x_6,z_6)}$, yields the orthogonal decomposition
\begin{align}
\Sym^6(\mathbb{C}^2)=\mathcal{H}_6^{(-1,-1)} \oplus \mathcal{M}_6^\perp,
\end{align}
where the orthogonal complement is taken within $\Sym^6(\mathbb{C}^2)$. In particular,
$\mathcal{M}_6=\mathcal{H}_6^{(-1,-1)}$. A direct calculation shows that the probability of projecting onto $\mathcal{M}_6$ is $\Mlin(\psi)/2$. Moreover, $\dim(\mathcal{M}_6)=1$, so conditioned on this outcome the input-state dependence is completely removed. We then show that a fixed stabilizer protocol involving three additional Pauli measurements converts the unique state spanning $\mathcal{M}_6$ into a $\CCZ$ state with probability $2/3$. Moreover, the accepted measurement branches differ only by transversal Clifford actions of the form $C^{\otimes 3}$ or $C^{\otimes 6}$. Multiplying the two success probabilities, we get the achievability result.

To prove the converse, note that by the argument of the previous paragraph, any exact successful concentration branch corresponds to a Kraus operator $K_r$ with input support contained in $\mathcal{M}_6$. A simple argument using the polynomial representation $\Phi \to p_{\Phi}$ shows that $\dim(\mathcal{M}_k)=0$ when $k<6$, hence no successful branch exists. On the other hand, when $k=6$, as discussed above, there is only one vector $\ket{A_6}$ spanning $\mathcal{M}_6$. In the appendix we show that the optimal conversion probability starting from $\ket{A_6}$ is indeed $2/3$, hence our protocol is optimal. A no-go via the stabilizer nullity further shows that at most one $\CCZ$ state can be obtained.
\end{proof}
\end{theorem}

This protocol fits naturally into the Clifford commutant framework \cite{Gross_2021,Bittel_2026}, since the effective measurement on $\psi^{\otimes 6}$ is $\Pi_{\mathcal{M}_6}$ and it lies in the six-copy commutant. Restricted to the symmetric subspace, the first efficiently measurable Clifford-invariant operator is $\Omega_6=(I+X^{\otimes 6}+Y^{\otimes 6}+Z^{\otimes 6})/2$ \cite{Bittel_2026}. In our case, the first measurement in the protocol corresponds to the projector
\begin{align}
\Tilde{\Omega}_6 \coloneqq \frac{I-\Omega_6}{2} = \frac{1}{4}\left(I-X^{\otimes 6}-Y^{\otimes 6}-Z^{\otimes 6}\right),
\end{align}
and one has
\begin{align}
\operatorname{Tr} (\psi^{\otimes 6} \Tilde{\Omega}_6 ) = \frac{1}{2} \Mlin(\psi).
\end{align}

 \cref{thm:six-copy-optimality} can also be understood through the finite-copy stabilizer-testing results of Ref.~\cite{Bittel_2026}. There, no nontrivial test distinguishing stabilizer from non-stabilizer states exists with up to five copies, while at six copies the POVM element $\Tilde{\Omega}_6$ yields the optimal test. This mirrors our six-copy threshold: below six copies, the Clifford commutant cannot resolve the non-stabilizer sector, precluding universal concentration, whereas at six copies $\Tilde{\Omega}_6$ detects precisely the invariant sector exploited by our protocol.

 Let us now turn to the eight-copy case.

\begin{theorem}[Eight-copy protocol. Informal version of \cref{thm:eightac}] \label{thm:eight-copy-optimality} There exists an universal exact $\CCZ$ magic-state concentration protocol $\Lambda_8$ acting on $k=8$ copies of
$\psi$ preparing an exact $\CCZ$ state with success probability
\begin{align}\label{eq:prob8}
\Pr_{\Lambda_8}(\psi^{\otimes 8} \to \CCZ) &= \frac{2}{3}\Mlin(\psi).
\end{align}
\begin{proof}[Proof sketch]
In this case, the joint Pauli measurement $X^{\otimes 8},Z^{\otimes 8}$ does not project by itself onto $\mathcal{M}_{8}$ as in the six-copy case. However, in the appendix we show that it decomposes $\mathcal{M}_{8}$ into the three orthogonal syndrome sectors
\begin{align}
\mathcal{M}_{8} = \mathcal{M}_{8}^{(+1,-1)} \oplus \mathcal{M}_{8}^{(-1,+1)} \oplus \mathcal{M}_{8}^{(-1,-1)}, 
\end{align}
where $\mathcal{M}_{8}^{(x_8,z_8)} \coloneqq \mathcal{M}_{8} \cap \mathcal{H}_8^{(x_8,z_8)}$. The syndrome sector $(+1,+1)$ is instead orthogonal to $\mathcal{M}_{8}$. Furthermore, each sector is one-dimensional. While in principle the input $\psi^{\otimes 8}$ belongs to the larger symmetric subspace $\Sym^8(\mathbb{C}^2)$, we show that subsequent Pauli measurements remove components orthogonal to $\mathcal{M}_{8}$ by post-selection, and hence the overall effect is to filter each syndrome sector onto its one-dimensional component $\mathcal{M}_{8}^{(x_8,z_8)}$. 
Every accepted branch then yields a state Clifford-equivalent to $\CCZ$. The total accepted fraction of the input weight in $\mathcal M_8$ is $4/7$.
The total success probability therefore is
\begin{align}
\Pr_{\Lambda_8}(\psi^{\otimes 8} \to \CCZ) = \frac{4}{7} \Tr\left( \Pi_{\mathcal{M}_{8}} \psi^{\otimes 8}\right) = \frac{2}{3} \Mlin(\psi),
\end{align}
thus concluding the proof of achievability.
\end{proof}
\end{theorem}

Both protocols are structurally Clifford invariant: their success probabilities satisfy $\Pr(\psi^{\otimes k}\to\CCZ)=\Pr((C\psi C^\dagger)^{\otimes k}\to\CCZ)$ for every $C\in\mathcal{C}_1$. Consequently, a Clifford twirl $\psi^{\otimes k}\to\mathbb{E}_{C\sim\mathcal{C}_1}[(C\psi C^\dagger)^{\otimes k}]$ leaves the success probability unchanged.
While the eight-copy protocol is not yet known to be optimal, the representation theory of the Clifford group reveals a much broader structure. Up to nine copies, every universal exact protocol to any fixed non-stabilizer target $\tau$ is constrained by $\Mlin$ up to a constant, while Clifford invariance rigidly fixes the success probability to $\Pr(\psi^{\otimes k}\to\tau)=c_k\,\Mlin(\psi)$, with $c_k$ independent of the input. Thus, throughout this low-copy regime, $\Mlin$ emerges as the unique parameter controlling exact magic concentration. This is consistent with the Clifford commutant structure \cite{Bittel_2026}: at seven copies, no new invariant operator appears beyond $\Tilde{\Omega}_6$, which acts non-trivially on only six tensor factors. Formally, we have

\begin{proposition}[Clifford-invariant rigidity. Informal of \cref{prop:opt6to9,prop:rigid}]
\label{prop:rigcliff}
Let $6\leq k\leq9$. The success probability of any universal Clifford-invariant stabilizer protocol in converting $\ket{\psi}^{\otimes k}$ into any fixed non-stabilizer state $\ket{\tau}$ is necessarily of the form
\begin{align}
\Pr(\psi^{\otimes k}\to\tau)
=
c_k\,\Mlin(\psi),
\label{eq:343}
\end{align}
for some constant $c_k$ independent of $\psi$. Furthermore \cref{eq:343} becomes only an upper bound for any universal (not necessarily Clifford-invariant) stabilizer protocol with $c_k\leq 7(k-5)/2(k+1)$.
\end{proposition}

In particular, this gives a $7\Mlin(\psi)/6$ upper bound when $k=8$. 
Starting at $k=10$, additional nontrivial and efficiently measurable Clifford invariants, independent of $\Mlin$, appear within $\mathcal M_k$~\cite{Bittel_2026}. This shows that the scaling achieved by our protocols is essentially optimal within this low-copy regime.

\parhead{Implications.---} We now discuss some direct implications of our results. While converting an arbitrary number of copies of a pure state into exact $T$ via stabilizer operations is impossible \cite{Beverland_2020}, our results, together with the catalytic procedure of Ref.\,\cite{GidneyFowler2019} directly imply the following.

\begin{corollary}[Universal exact catalytic $T$-state conversion]
\label{cor:catalytic-t}
Assume access to a single $T$-state catalyst. Then the six- and eight-copy protocols induce exact catalytic $T$-state conversions, with net balance
\begin{align}
\Pr_{\Tilde{\Lambda}_{6}}(\psi^{\otimes 6}\to_{\mathrm{cat}} T^{\otimes 2}) &= \frac{1}{3}\Mlin(\psi), \\
\Pr_{\Tilde{\Lambda}_{8}}(\psi^{\otimes 8}\to_\mathrm{cat} T^{\otimes 2}) &= \frac{2}{3}\Mlin(\psi).
\end{align}
\begin{proof}
Compose each successful exact-$\ket{\CCZ}$ branch with the deterministic catalytic stabilizer conversion \cite{GidneyFowler2019} $\ket{\CCZ}\ket{T} \longmapsto \ket{T}^{\otimes 2}\ket{T}$.
Since the catalytic conversion succeeds deterministically and returns the catalyst exactly, it leaves the success probabilities unchanged and produces
two net $T$ states per successful $\ket{\CCZ}$ output.
\end{proof}
\end{corollary}

Repeating our finite-copy protocols over independent blocks immediately yields asymptotic concentration rates. We define $R(\psi \to\CCZ)$ as the supremum of achievable exact $\CCZ$ rates over protocol families whose total success probability converges to one as $k\to\infty$ (see \cref{sec:protdef} for the formal definition). The catalytic rate $R(\psi \to_{\mathrm{cat}} T)$ is defined similarly.  We then obtain the following lower bounds

\begin{corollary}[Asymptotic rates] \label{cor2:rates} The eight-copy protocol provides the following achievability statements on asymptotic magic state concentration protocols.
\begin{align}
\begin{aligned}
R(\psi \to \CCZ) &\geq \frac{1}{12} \Mlin(\psi), \\
R(\psi \to_{\mathrm{cat}} T) &\geq \frac{1}{6} \Mlin(\psi).
\end{aligned}
\end{align}
\begin{proof}
Split the $k$ copies into $N=\lfloor k/8\rfloor$ blocks and apply the eight-copy protocol independently. Then the number of $\CCZ$ in the output is given by the binomial distribution $S_k\sim\mathrm{Bin}(N,2\Mlin(\psi)/3)$ , so $S_k/k\xrightarrow{\Pr}\Mlin(\psi)/12$, proving the first bound. Catalytic conversion of each successful $\CCZ$ into two net $T$ states doubles the yield and gives the second.
\end{proof}
\end{corollary}

One might asks whether the asymptotic rates of \cref{cor2:rates} are by any chance tight. In the next theorem, we show that this is indeed the case up to logarithmic factors. The result is obtained using simple bounds based on the relative entropy of magic \cite{Rubboli_2024,HORODECKI_2012,Veitch_2014}.

\begin{proposition}[Asymptotic upper bound. Informal of \cref{prop:globalbound} ]\label{prop:mainlow}
    For any pure non-stabilizer qubit state $\psi$
    \begin{align}
    R(\psi \to \CCZ) \leq C \Mlin(\psi) \log_2\left( \frac{1}{\Mlin(\psi)}\right),
    \end{align}
    with constant $C < 1.81$.
\end{proposition}

If we restrict the allowed protocols to independent universal many-to-one blocks, we obtain a substantially tighter characterization. Denoting by $R_{\rm univ}^{(1)}(\psi\to\CCZ)$ the optimal asymptotic rate achievable by combining blocks that act on fresh copies of $\psi$ and either produce one exact $\CCZ$ state or fail and are discarded, we recover the optimal scaling up to constant factors.

\begin{proposition}[Asymptotic universal many-to-one optimality. Informal of \cref{thm:univone}]\label{prop:many-to-one}
For every pure qubit state $\psi$ and every fixed non-stabilizer target state $\tau$,
\begin{align}
R_{\rm univ}^{(1)}(\psi\to\tau)
\le
\frac{1}{2}M_{\rm lin}(\psi).
\end{align}
\end{proposition}
For $\tau=\CCZ$, our universal eight-copy protocol achieves $M_{\rm lin}(\psi)/12$, and therefore
\begin{align}
\frac{1}{12}M_{\rm lin}(\psi)
\le
R_{\rm univ}^{(1)}(\psi\to\CCZ)
\le
\frac{1}{2}M_{\rm lin}(\psi).
\end{align}

Hence, the eight-copy protocol is asymptotically optimal up to constant factors within the universal many-to-one setting. More generally, $\Mlin$ acquires a direct asymptotic operational meaning: it fixes the optimal concentration-rate scaling up to logarithmic factors, and the simple protocols of \cref{thm:six-copy-optimality,thm:eight-copy-optimality} achieve this scaling.
Finally, we also have

\begin{corollary}[Universal quantum computation from unknown pure qubit magic]\label{cor:uqc}
Any unknown pure non-stabilizer qubit state enables universal quantum computation through exact CCZ-gate injection by a fixed stabilizer protocol with finite state-dependent overhead.
\begin{proof}
The six-copy protocol succeeds with probability $\Mlin(\psi)/3>0$ for every pure non-stabilizer $\psi$. Repetition therefore produces exact $\ketccz$ states at finite expected cost, while Clifford operations together with $\CCZ$ injection are universal \cite{shi2002toffolicontrollednotneedlittle}.
\end{proof}
\end{corollary}

In comparison with Reichardt's universality results \cite{Reichardt2009,Reichardt2005}, where the distillation strategy may depend on the input Bloch-sphere direction, our protocol is fixed and requires neither prior alignment nor knowledge of that direction.

The proofs of \cref{thm:six-copy-optimality,thm:eight-copy-optimality} extend directly to arbitrary states supported on the symmetric subspace. There, the measurement effectively implement $(2/3)\Pi_{\mathcal M_6}$ and $(4/7)\Pi_{\mathcal M_8}$, so any such mixed state is mapped upon success to the same exact $\CCZ$ output, with probability set by its weight on $\mathcal M_k$. This gives the following result.

\begin{corollary}[Concentration from symmetric mixed states]\label{cor4main}
Let $\rho_k$ be any mixed state supported on $\Sym^k(\mathbb C^2)$. For $k=6,8$, the corresponding protocols produce an exact $\CCZ$ state with probabilities
\begin{align}\label{eqmaincor4main}
\Pr_{\mathrm{opt}}(\rho_6\to\CCZ)
&=\frac23\Tr\!\left(\Pi_{\mathcal M_6}\rho_6\right),\\
\Pr_{\Lambda_8}(\rho_8\to\CCZ)
&=\frac47\Tr\!\left(\Pi_{\mathcal M_8}\rho_8\right).
\end{align}
\end{corollary}
Note that in general, it is not possible to project onto the symmetric subspace with a stabilizer protocol, therefore, \cref{cor4main} does not extend trivially to mixed states.

\parhead{Discussion.---} We introduced exact universal concentration of an unknown pure qubit magic state into a fixed target, focusing on $\CCZ$. While our explicit protocols target $\CCZ$, several of the underlying structural results apply to arbitrary fixed non-stabilizer targets. Unlike standard magic-state distillation \cite{BravyiKitaev2005,Reichardt2005,Reichardt2009}, our protocols require neither prior knowledge nor alignment of the input Bloch-sphere direction: a single stabilizer procedure applies universally, with Clifford-invariant success probability. They also admit a stabilizer-code interpretation: the measurements of $X^{\otimes 6},Z^{\otimes 6}$ and $X^{\otimes 8},Z^{\otimes 8}$ define parity-code syndromes, while subsequent stabilizer constraints further restrict the code space before decoding. Geometrically, these codes filter the stabilizer-orthogonal symmetric subspaces $\mathcal{M}_k$, isolating components that stabilizer operations map exactly to $\ket{\CCZ}$, without using transversal non-Clifford logical gates. The six-copy protocol is optimal, while optimality at eight copies remains open. Natural directions for future work include extending the protocols to arbitrary copy numbers, determining optimal asymptotic rates, and generalizing the framework to qudits and other target resources.

\let\savedaddcontentsline\addcontentsline
\def\addcontentsline#1#2#3{}

\section*{Acknowledgements}

The authors acknowledge insightful discussions with Salvatore F.E. Oliviero.
J.R. was supported by the BMFTR (PasQuops, QSolid, MuniQC-Atoms), the ML4Q excellence cluster, the Munich Quantum Valley (K-8), Berlin Quantum, the DFG (SPP 2514, project ID 563402549), and the European Research Council (DebuQC).

\makeatletter
\@bibdataout@aps
\makeatother
\bibliographystyle{apsrev4-2}
\bibliography{biblio}

\let\addcontentsline\savedaddcontentsline

\clearpage
\appendix
\onecolumngrid
\renewcommand{\tocname}{Supplemental Material}
\tableofcontents

\section{Preliminaries}

\subsection{Stabilizer formalism}\label{sec:stabform}
In this section, we briefly summarize the stabilizer formalism. A more detailed overview can be found in Ref.\,\cite{gottesman2009introductionquantumerrorcorrection}.
In the manuscript, the $n$-qubit Pauli group $\mathcal{P}_n$ is composed of all $n$-fold tensor products of $I,X,Y,Z$ with phases $\pm1,\pm i$. We also denote with $\widehat{\mathcal{P}}_n\coloneqq \mathcal{P}_n/\{\pm I,\pm iI\}\simeq\{I,X,Y,Z\}^{\otimes n}$ the Pauli group modulo phase. When we write a quantity in terms of elements of $\widehat{\mathcal{P}}_n$, the value of the phase will not matter.
Given an $n$-qubit state $\ket{\phi}$, its stabilizer group is defined as
\begin{align}
\mathcal{S}(\phi) \coloneqq \{ P\in \mathcal{P}_n : \, P\ket{\phi}=\ket{\phi}\}.
\end{align}
This definition immediately implies that $\mathcal{S}(\phi)$ is an Abelian subgroup of $\mathcal{P}_n$, and $-I\notin \mathcal{S}(\phi)$, hence the corresponding Paulis can appear only with phases $\pm 1$. If $\mathcal{S}(\phi)$ is generated by $r$ independent commuting Paulis $S_1,...,S_r \in \mathcal{P}_n$, we write $\mathcal{S}(\phi) = \langle S_1,...,S_r \rangle$. Consequently, since the generators commute and $S_i^2=I$, we have $|\mathcal{S}(\phi)|=2^r$. We can also reverse the definition: given an Abelian subgroup $\mathcal{S} \subseteq \mathcal{P}_n$ such that $-I\notin \mathcal{S}$, we define the corresponding stabilizer code as the subspace $\mathcal{T}$ of the $n$-qubit states stabilized by $\mathcal{S}$:
\begin{align}
\mathcal{T}(\mathcal{S}) \coloneqq \{ \ket{\phi} \in (\mathbb{C}^2)^{\otimes n} : P\ket{\phi}=\ket{\phi},\, \forall P \in \mathcal{S} \}.
\end{align}
The corresponding projector onto $\mathcal{S}$ is easily written as
\begin{align}
\Pi_{\mathcal{S}} = \frac{1}{2^r} \prod_{i=1}^r (I + S_i),
\end{align}
Taking the trace on both sides then gives the dimension of the resulting subspace as $\dim(\mathcal{T}(\mathcal{S}))=2^{n-r}$. This also shows that $\mathcal{S}$ Abelian and $-I\notin \mathcal{S}$ are necessary and sufficient conditions in order for $\mathcal{T}$ to be non-trivial.
A state $\phi$ is said to be a stabilizer state whenever it is fully specified by its stabilizer group $\mathcal{S}$, that is $n=r$, hence $\dim(\mathcal{T}(\mathcal{S}))=1$. Otherwise, we say the state is non-stabilizer.
For a single-qubit state $\ket{\psi}$ (in the manuscript, we will adopt this different letter when referring to a single-qubit state), the convex hull of stabilizer states forms an octahedron described by the six vertices
\begin{align}
\mathrm{Stab}_1 \coloneqq \{ \ket{0}, \ket{1}, \ket{+}, \ket{-}, \ket{+i},\ket{-i} \},
\end{align}
where $\ket{\pm}=(\ket{0}\pm\ket{1})/\sqrt{2}$ and $\ket{\pm i}=(\ket{0}\pm i\ket{1})/\sqrt{2}$.
We denote by $\mathcal{C}_n\coloneqq\{U\in\mathcal{U}(2^n): \,U\mathcal{P}_nU^\dagger=\mathcal{P}_n\}$ the $n$-qubit Clifford group, by $H$ the Hadamard gate, by $\mathrm{CNOT}_{i\to j}$ the controlled-NOT gate with control qubit $i$ and target qubit $j$, and by $S$ the phase gate. $H,\mathrm{CNOT}$ and $S \coloneqq \mathrm{diag}(1,i)$ generate the $n$-qubit Clifford group.
We say that two states are Clifford equivalent, indicated with $\simeq_{\mathcal{C}}$, if they are connected by a Clifford unitary. We write $\simeq$ when they are equivalent up to a global phase.
A stabilizer protocol consists of a quantum circuit built from Clifford unitaries, preparation of pure stabilizer auxiliary states, Pauli measurements, classical feed-forward, and discarding of subsystems. Classical randomness may also be used.
We also define the stabilizer nullity as
\begin{align}
\nu(\phi) = n - \log_2 |\mathcal{S}(\phi)| = n - r.
\end{align}
A stabilizer state therefore has $\nu(\phi)=0$. A state with no nontrivial Pauli stabilizer has instead $\nu(\phi)=n$. By a Clifford change of basis, any pure state with $r$ independent Pauli stabilizers, is equivalent to
\begin{align}
\ket{\phi} \simeq_{\mathcal{C}} \ket{0}^{\otimes r}\otimes \ket{\omega},
\label{eq:nullc}
\end{align}
for some $n-r=\nu(\phi)$-qubit state $\ket{\omega}$.
This is due to the fact that any rank-$r$ Stabilizer group can be brought into the normal generators form $\langle Z_1,...,Z_r\rangle$ by an appropriate Clifford unitary. 
In general, under a Clifford unitary $C \in \mathcal{C}_n$, the stabilizer group of an arbitrary state $\ket{\phi}$, when $\ket{\phi'}\coloneqq C\ket{\phi}$, changes as $\mathcal{S}(\phi') = C\mathcal{S}(\phi)C^\dag$.
In the following, it will be useful to recap how does the stabilizer group change after a Pauli measurement.
Suppose a pure state $\phi$ has stabilizer group $\mathcal{S} = \langle S_1,...,S_r\rangle$. Now measure any Hermitian Pauli $P$ with outcome $s=\pm 1$. The post-measurement state is
\begin{align}
\ket{\phi_s} = \frac{\Pi_s \ket{\phi}}{\|\Pi_s \ket{\phi}\|}, \quad \Pi_s = \frac{I+sP}{2}.
\end{align}
Three things can happen. We will refer to these measurements as type (i)-(iii) in the rest of the manuscript.
\begin{enumerate}[label=(\roman*)]
    \item If $\pm P\in \mathcal{S}$, then the state remains unchanged, the outcome is deterministic and the stabilizer group $\mathcal{S}(\phi)$ does not change.
    \item If $PS=-SP$ for at least one $S\in \mathcal{S}$, then, up to replacing generators, we have $\{P,S_1\}=0$, and $[P,S_i]=0$ for $i=2,...,r$. Then, $\langle P\rangle_\phi=0$, and 
    \begin{align}
    \ket{\phi_s} = \frac{(I+sP)}{\sqrt{2}}\ket{\phi}.
    \end{align}
    In particular, this implies that $PS_1$ is anti-Hermitian, and $\ket{\phi_s}=U_s \ket{\phi}$ with $U_s$ Clifford unitary given by
    \begin{align}
    U_s \coloneqq \frac{I+sPS_1}{\sqrt{2}} = e^{\frac{\pi}{4}sPS_1},
    \end{align}
    Since $\exp(\frac{\pi}{4}Q)$ is Clifford if $Q^\dag=-Q$ and $Q \in \mathcal{P}_n$.
    Hence the branch is Clifford-equivalent to $\phi$, i.e. $\phi \simeq_{\mathcal{C}}\phi_s$. Tracking the generators, the only one that changes is $S_1 \to sP$. This implies that the new stabilizer group is $\langle sP,S_2,...,S_{r} \rangle = \mathcal{S}(\phi_s)$, and in particular $\nu(\phi_s)= \nu(\phi)$.
    \item If $P$ commutes with any $S\in \mathcal{S}$, but $\pm P\notin \mathcal{S}$, then all the older stabilizers survive, and we can add at least one more independent one. Hence we have at least $r+1$ stabilizer generators, therefore $\nu(\phi_s)\leq\nu(\phi)-1$.  
\end{enumerate}
We also recap briefly the binary notation for Paulis. For $x,z\in \mathbb{F}_2^{n}$, set $a \coloneqq (x,z)$ and define
\begin{align}
W_a \coloneqq i^{x\cdot z} X^x Z^z,
\label{eq:binary}
\end{align}
so that we identify $(0,0)\leftrightarrow I, (1,0)\leftrightarrow X, (0,1)\leftrightarrow Z, (1,1)\leftrightarrow Y$, and the dot product $x\cdot z$ is intended modulo 2. Equivalently, $x\cdot z$ is the parity of the number of $Y$'s.
The action of $W_a$ on the computational basis is given by
\begin{align}
W_a \ket{y} = i^{x\cdot z} (-1)^{z\cdot y} \ket{y+x},
\end{align}
with addition modulo 2.
The associated symplectic form, if $b=(x',z')$ is
\begin{align}
[a,b] \coloneqq x \cdot z' + z \cdot x' = (x+x')\cdot(z+z') + x \cdot z + x' \cdot z'.
\end{align}
Sometimes we will directly indicate $q(a) \equiv x \cdot z$. Then $[a,b]= q(a+b)+q(a)+q(b)$. $W_a,W_b$ commute iff $[a,b]=0$.

\subsection{Magic state concentration protocols}\label{sec:protdef}
In this section we formally define the finite-copy exact magic state concentration task, as well as its asymptotic extension. Our definitions follow the respective versions in entanglement theory \cite{hayashi2002universaldistortionfreeentanglementconcentration}. We start with the finite-copy task.

\begin{definition}[Exact $\CCZ$ magic-state concentration protocol]
\label{def:exact}
An exact $(k,m)$ $\CCZ$ magic-state concentration protocol is a stabilizer
protocol $\Lambda$ acting on $k$ input copies, with a fixed set of successful
outcomes $\mathcal R_{\mathrm{succ}}$, such that, for every pure single-qubit
state $\psi$ and every $r\in\mathcal R_{\mathrm{succ}}$,
\begin{align}
K_r\ket{\psi}^{\otimes k}=c_r(\psi)\ket{\CCZ}^{\otimes m},
\end{align}
where $K_r$ includes the outcome-dependent Clifford correction associated
with branch $r$, and $c_r(\psi)\in\mathbb C$. The corresponding success
probability is
\begin{align}
\Pr_{\Lambda}\!\left(\psi^{\otimes k}\to\CCZ^{\otimes m}\right)
=\sum_{r\in\mathcal R_{\mathrm{succ}}}|c_r(\psi)|^2.
\end{align}
The protocol is universal if its circuit, measurements, successful outcomes,
and outcome-dependent Clifford corrections are all independent of
$\psi$.
\end{definition}

We can now define the asymptotic version as follows. While we keep the output exact, we will allow for the success probability of achieving a given rate to approach one asymptotically.

\begin{definition}[Universal exact asymptotic $\CCZ$ magic-state concentration]
\label{def:asy}
We say a rate $r(\psi)\geq0$ is universally achievable if there exists a sequence of universal stabilizer protocols $\{\Lambda_k\}_k$ such that,
for every pure single-qubit state $\ket{\psi}$, $\Lambda_k$ outputs a random
number $S_k$ of exact $\CCZ$ states satisfying
\begin{align}
\Pr_{\psi}\!\left[ \frac{S_k}{k}\geq r(\psi)-\varepsilon \right] \xrightarrow[k\to\infty]{} 1, \qquad \forall\,\varepsilon>0.
\end{align}
The optimal universal exact asymptotic $\CCZ$ concentration rate $R(\psi\to\CCZ)$ is the supremum over all universally achievable rates $r(\psi)$.
\end{definition}

\begin{remark}
This definition is stronger than the usual asymptotic notion
\cite{Rubboli_2024,Wilde_2016}, where the target state may be prepared approximately.
Indeed, if $r$ is universally achievable, then for any $R<r(\psi)$ and
$m_k\coloneqq\lfloor Rk\rfloor$,
\begin{align}
p_k(\psi)\coloneqq\Pr_{\psi}[S_k\geq m_k]\longrightarrow1.
\end{align}
Keeping any $m_k$ successful outputs and assigning an arbitrary state
$\omega_k$ otherwise gives
\begin{align}
\sigma_k
=
p_k(\psi)\ketbra{\CCZ}^{\otimes m_k}
+
(1-p_k(\psi))\omega_k,
\end{align}
and therefore
\begin{align}
D\!\left(\sigma_k,\ketbra{\CCZ}^{\otimes m_k}\right)
\leq1-p_k(\psi)\longrightarrow0.
\end{align}
\end{remark}

\subsection{Linearized order-three stabilizer Rényi entropy}
\label{app:linstab}
In the manuscript, we will consider stabilizer protocols acting on multiple copies of a fixed single-qubit state $\ket{\psi} = a\ket{0} + b\ket{1}$, with $a,b\in\mathbb{C}$ and $|a|^2+|b|^2=1$. In the following, we will make use of the simple Clifford-unitary equivalence
\begin{align}
\ketccz = H_1 \mathrm{CNOT}_{1\to3} \mathrm{CNOT}_{1\to2} \ket{\chi}, 
\label{eq:cleq}
\end{align}
where
\begin{align}
\ket{\chi} \coloneqq \frac{1}{2}(\ket{000}+\ket{001}+\ket{010}+\ket{100}),
\end{align}
and
\begin{align}
\ketccz \coloneqq \frac{1}{\sqrt{8}} \sum_{x_1,x_2,x_3\in\{0,1\}} (-1)^{x_1x_2x_3} \ket{x_1x_2x_3}.
\end{align}
For $k$ copies of a qubit, we also define the $k$-copy symmetric subspace as
\begin{align}
\Sym^k(\mathbb{C}^2) \coloneqq \left\{ \ket{\Phi}\in(\mathbb{C}^2)^{\otimes k} : R_\pi\ket{\Phi}=\ket{\Phi}, \ \forall \pi\in S_k \right\},
\end{align}
where $S_k$ is the symmetric group of $k$ elements and $R_\pi$ are the associated permutation actions acting as $R_\pi \ket{x_1,\ldots,x_k} \coloneqq \ket*{x_{\pi^{-1}(1)},\ldots,x_{\pi^{-1}(k)}}$, where $x_i \in \{0,1\}$. 
In particular, $\dim(\Sym^k(\mathbb{C}^2))=k+1$. We also define the Dicke states $\ket*{D_w^k}$, corresponding to the uniform superposition of $k$-bit strings of weight $w$
\begin{align}
\ket*{D_w^k}
\coloneqq
\binom{k}{w}^{-1/2}
\sum_{\substack{x\in\{0,1\}^k\\ |x|=w}}
\ket{x}, \quad w = 0,...,k.
\label{eq:dicke}
\end{align}
These form an orthonormal basis of $\Sym^k(\mathbb{C}^2)$.
We now define some useful measures of magic.
The order-$\alpha$ stabilizer purity of $\phi$ is defined for any $\alpha>1$ as
\begin{align}
P_\alpha(\phi) \coloneqq 2^{-n} \sum_{P\in \widehat{\mathcal{P}}_n} \left|\expval{P}{\phi}\right|^{2\alpha},
\end{align}
while the stabilizer R\'enyi entropy of order $\alpha$ and its linearized version are defined as \cite{LeoneOlivieroHamma2022}
\begin{align}
M_\alpha(\phi) \coloneqq \frac{1}{1-\alpha} \log_2 P_\alpha(\phi), \quad M_\alpha^{\mathrm{lin}}(\phi) \coloneqq 1 - P_\alpha(\phi).
\end{align}
$M_\alpha$ is faithful, Clifford invariant and additive. For integer $\alpha\geq2$, it is also a monotone under stabilizer protocols, while the corresponding linearized quantities satisfy strong monotonicity \cite{LeoneOlivieroHamma2022,LeoneBittel2024}. 
Monotonicity of $M_\alpha$ means $M_\alpha(\phi) \geq M_\alpha(\Lambda(\phi))$ for any stabilizer protocol $\Lambda$, while strong monotonicity of the linearized version implies that if a stabilizer protocol applied to $\ket{\phi}$ produces a pure state $\ket{\phi_r}$ with probability $p_r$, then
\begin{align} 
M_\alpha^{\mathrm{lin}}(\phi) \geq \sum_r p_r M_\alpha^{\mathrm{lin}}(\phi_r).
\end{align}
In this manuscript, we will only deal with the  linearized order-three stabilizer Rényi entropy.
For a pure single-qubit state $\ket{\psi}$, the Bloch vector is defined as $x = \langle X \rangle_\psi$, $y = \langle Y \rangle_\psi$, and $z = \langle Z \rangle_\psi$, with $x,y,z\in\mathbb{R}$ and $x^2+y^2+z^2=1$. We also have the relations $x = 2\mathrm{Re}(a^*b)$, $y = 2\mathrm{Im}(a^*b)$, $z = |a|^2-|b|^2$, giving
\begin{align}
|a^2-b^2|^2 = 1-x^2, \quad |a^2+b^2|^2 = 1-y^2, \quad |2ab|^2 = 1-z^2.
\label{eq:basic1}
\end{align}
The order-three stabilizer purity for a single qubit then reads $\Pthree(\psi)=(1+x^6+y^6+z^6)/2$. We also have 
\begin{align}
\begin{aligned}
\Mlin(\psi) = 6|f_6(a,b)|^2, \quad f_6(a,b) \coloneqq ab(a^4-b^4).
\end{aligned}
\label{eq:f6}
\end{align}
\cref{eq:f6} follows from
\begin{align}
\begin{aligned}
\Mlin(\psi) &= \frac{1}{2}(1-x^6-y^6-z^6) \\
&\overset{\text{(i)}}= \frac{3}{2} (1-x^2)(1-y^2)(1-z^2) \\
&\overset{\text{(ii)}}= \frac{3}{2} |2ab|^2 |a^2+b^2|^2 |a^2-b^2|^2 \\
&= 6 |ab|^2 |a^4-b^4|^2 \\
&= 6|f_6(a,b)|^2,
\end{aligned}
\label{eq:der1}
\end{align}
where in (i) we used the identity
\begin{align}
(u+v+w)^3 - u^3 -v^3 - w^3 = 3(u+v)(v+w)(w+u),
\end{align}
setting $u=x^2,v=y^2,w=z^2$, and in (ii) we used \cref{eq:basic1}.

We now make precise the polynomial argument used in
\cref{thm:six-copy-optimality,thm:eight-copy-optimality} in the main text. The idea is to first decompose any $\ket{\Phi} \in \Sym^{k}(\mathbb{C}^2)$ in the Dicke basis (see \cref{eq:dicke}) as
\begin{align}
\ket{\Phi} = \sum_{w=0}^k c_w \ket*{D_w^k}.
\end{align}
We then consider the $k$-fold tensor product of state $\ket{\psi}$ and express it in the same basis
\begin{align}
\ket{\psi}^{\otimes k} = \sum_{w=0}^k \sqrt{\binom{k}{w}} a^{k-w} b^w \ket*{D_w^k}.
\end{align}
and then associate the following degree-$k$ homogeneous polynomial to $\ket{\Phi}$.
\begin{align}
p_{\Phi}(a,b) \coloneqq \braket{\Phi}{\psi}^{\otimes k} = \sum_{w=0}^k c_w^* \sqrt{\binom{k}{w}} a^{k-w} b^w.
\label{eq:isom}
\end{align}
The equation above then constitutes a linear isomorphism between the space of degree-$k$ homogeneous polynomials and the dual of the symmetric subspace for $k$ qubits.
\begin{align}
\Sym^{k}(\mathbb{C}^2)^* \cong \mathbb{C}[a,b]_k,
\end{align}
where
\begin{align}
\mathbb C[a,b]_k \coloneqq \left\{\sum_{w=0}^{k} d_w\,a^{k-w}b^w : d_w\in\mathbb C\right\}.
\end{align}
We now derive a simple characterization of all homogeneous polynomials vanishing on the six pure qubit stabilizer states. The linearized order-three stabilizer Rényi entropy $\Mlin$ will essentially capture this polynomial constraint.
More precisely, we define the vector space of degree-$k$ homogeneous vanishing polynomials on $\mathrm{Stab}_1$ as
\begin{align}
\mathcal{I}(\mathrm{Stab_1})_k \coloneqq \{ p \in \mathbb C[a,b]_k:\,\, p(\mathrm{Stab_1})=0 \}.
\end{align}
Note that homogenity implies that this definition does not depend on the normalization chosen for $a,b$.
The following lemma is then a direct consequence of the definition.

\begin{lemma}[Polynomials vanishing on the vertices of the stabilizer octahedron]\label{lem:homg}
We have
\begin{align}
\mathcal{I}(\mathrm{Stab_1})_k =
\begin{cases}
\{0\}, \quad &k<6, \\
f_6\,\mathbb C[a,b]_{k-6}, \quad &k\geq 6,
\end{cases}
\end{align}
where $f_6(a,b)=ab(a^4-b^4)$.

\begin{proof}
Let $p(a,b)\in\mathcal{I}(\mathrm{Stab_1})_k$. Since $p(1,0)=p(0,1)=0$, the factors $b$ and $a$, respectively, divide $p$, and therefore $p(a,b)=ab\,r(a,b)$, with $r$ homogeneous of degree $k-2$. The remaining four stabilizer states imply $r(1,\pm1)=r(1,\pm i)=0$. Hence, by the factor theorem, $r(1,z)$ is divisible by
\begin{align}
(z-1)(z+1)(z-i)(z+i)=z^4-1.
\end{align}
Hence $r(1,z) = (z^4-1)s(z)$, with $s$ polynomial of degree at most $k-6$. Now, for $a\neq0$, define $q(a,b)\coloneqq a^{k-6} s(b/a)$. $q$ is homogeneous of degree $k-6$. Now, by homogeneity of $r$
\begin{align}
\begin{aligned}
r(a,b) &= a^{k-2} r\left(1, \frac{b}{a}\right) \\
&= a^{k-2} \left[ \left(\frac{b}{a}\right)^4 - 1\right] s\left( \frac{b}{a} \right) \\
&= \left( b^4 - a^4\right) a^{k-6} s\left( \frac{b}{a} \right) \\
&= \left( b^4 - a^4\right) q(a,b).
\end{aligned}
\end{align}
Thus, for $a\neq0$, up to re-defining $q$ by an overall sign
\begin{align}
p(a,b)=f_6(a,b)q(a,b).
\end{align}
Since they are polynomials, equality extends to $a=0$.
For $k<6$, no nonzero degree-$k$ homogeneous polynomial can contain the degree-six factor $f_6$, hence $\mathcal{I}(\mathrm{Stab_1})_k=\{0\}$. For $k\geq6$, the above proves $\mathcal{I}(\mathrm{Stab_1})_k\subseteq f_6\mathbb C[a,b]_{k-6}$, while the converse inclusion follows immediately from the fact that $f_6$ vanishes on all elements of $\mathrm{Stab_1}$.
\end{proof}
\end{lemma}

\section{Six- and eight-copy achievability}
In this section we formally derive the six- and eight-copy magic state concentration protocols. The main idea behind both protocols is to consider the stabilizer-orthogonal symmetric subspace
\begin{align}
\mathcal{M}_k \coloneqq \left\{ \ket{\Phi} \in \Sym^k(\mathbb{C}^2) : \, \braket{s^{\otimes k}}{\Phi} = 0,\, \forall s \in \mathrm{Stab_1}  \right\}.
\label{eq:mkdef}
\end{align}
Note that $\mathcal{M}_k$ is invariant under tensored Cliffords, namely
\begin{align}
\ket{\Phi} \in \mathcal{M}_k \iff C^{\otimes k}\ket{\Phi} \in \mathcal{M}_k, \quad \forall C \in \mathcal{C}_1.
\label{eq:cliffinv}
\end{align}
Furthermore, note any Kraus operator $K_r$ corresponding to some branch of a stabilizer protocol can only have support on $\mathcal{M}_k$, when restricted to the symmetric subspace, since for a stabilizer state, necessarily $K_r\ket{s}^{\otimes k}=0$.
This implies
\begin{align}
K_r \Pi_{\Sym^k(\mathbb{C}^2)} = K_r \Pi_{\mathcal{M}_k}.
\label{eq:proj}
\end{align}
The main strategy of the $k=6,8$ protocols is to project onto a suitable basis of $\mathcal{M}_k$, and then extract a fixed magic state. All procedures must be done using only stabilizer operations. The following lemma follows via the tools from the previous section, and will be useful to better place the structure of our protocols.

\begin{lemma}[Dimension of $\mathcal{M}_k$]\label{lem:dimmk} We have
\begin{align}
\dim(\mathcal{M}_k) = \begin{cases} 0, \quad &k<6, \\
k-5, \quad &\text{otherwise}.
\end{cases}  
\end{align}
\begin{proof}
Take any $\ket{\Phi}\in \mathcal{M}_k$. In particular, $\ket{\Phi} \in \Sym^k(\mathbb{C}^2)$.
Now, by \cref{eq:mkdef} and the linear isomorphism in  \cref{eq:isom}, $\ket{\Phi}\in$ $\mathcal{M}_k \iff p_{\Phi}(a,b) \in \mathcal{I}(\mathrm{Stab}_1)_k$, but from \cref{lem:homg} we know $\dim(\mathcal{I}(\mathrm{Stab}_1)_k) = \dim(\Sym^2(\mathbb{C}^2)) =k-5$ when $k\geq6$ and $\dim(\mathcal{I}(\mathrm{Stab}_1)_k)=0$ otherwise. Therefore $\dim (\mathcal{M}_k) = k-5$ when $k\geq6$ and $\dim(\mathcal{M}_k)=0$ otherwise. 
\end{proof}
\end{lemma}

\subsection{Six-copy protocol}\label{app:six}
We hereby describe in more detail the six-copy universal magic state concentration protocol sketched in \cref{thm:six-copy-optimality} in the main text. Schematically, the process is a chain of Stabilizer measurements and Clifford unitaries mapping
\begin{align}
\ket{\psi}^{\otimes 6} \longrightarrow \ket{A_6} \longrightarrow \{\ket*{G_4},\ket*{\overline{G}_4}\} \longrightarrow \ket{\chi} \longrightarrow \ket{\CCZ},
\end{align}
where only the first two arrows discard some outcomes, and
\begin{align}
\ket{A_6} \coloneqq \frac{\ket*{D_1^6} - \ket*{D_5^6}}{\sqrt{2}},
\end{align}
with $\ket*{G_4},\ket*{\overline{G}_4}$
to be defined later. Note that $\ket{A_6} \in \mathcal{M}_6$, since the Dicke states belong to the symmetric subspace, and $\ket{A_6}$ is orthogonal to all six single-qubit stabilizer states. Actually, it is the only vector spanning $\mathcal{M}_6$ since by \cref{lem:dimmk}, $\dim(\mathcal{M}_6)=1$.
More formally, we now show how to project onto $\ket{A_6}$ via stabilizer measurements, and extract the respective $\CCZ$ state. The protocol will then be exact and universal in the sense specified in \cref{def:exact}.
\begin{theorem}[Six-copy achievability]\label{thm:sixprob} There exists an explicit six-copy exact $\CCZ$ universal magic state concentration protocol $\Lambda_6$ with success probability
\begin{align}
\Pr_{\Lambda_6}(\psi^{\otimes 6} \to \CCZ) &= \frac{1}{3}\Mlin(\psi).
\end{align}
\begin{proof}
We start noting that $\ket{\psi}^{\otimes 6} \in \Sym^{6}(\mathbb{C}^2)$ and the Dicke states $\ket*{D_0^6},...,\ket*{D_6^6}$ form an orthonormal basis.
We now measure the two commuting observables $X^{\otimes 6},Z^{\otimes 6}$. We have $Z^{\otimes 6}\ket*{D_w^6} = (-1)^w\ket*{D_w^6}$, hence selecting outcome $-1$ restricts to the $3$-dimensional subspace $\mathrm{span}\{\ket*{D_1^6},\ket*{D_3^6},\ket*{D_5^6}\}$.
Instead $X^{\otimes 6}\ket*{D_w^6} = \ket*{D_{6-w}^6}$, therefore $X^{\otimes 6}$ exchanges $w=1$ with $w=5$ and stabilizes $w=3$. Post-selecting again on outcome $-1$, the only possible outcome is the vector $\ket{A_6}$.
We can expand the input state in the Dicke basis as
\begin{align}
\ket{\psi}^{\otimes 6} = \sum_{w=0}^6 \sqrt{\binom{6}{w}} a^{6-w} b^w \ket*{D_w^6}.
\end{align}
Computing the overlap gives
\begin{align}
\braket{A_6}{\psi}^{\otimes 6} = \frac{\sqrt{6}}{\sqrt{2}}(a^5 b - a b^5) = \sqrt{3}f_6(a,b).
\end{align}
Hence in this first measurement step
\begin{align}
\Pr(-1,-1) = 3 |f_6(a,b)|^2 = \frac{1}{2}\Mlin(\psi).
\end{align}
Post-selecting on this outcome, we measure the first two qubits in the computational basis and accept only outcomes $00,11$. Note that the choice of position is irrelevant due to permutation symmetry. This is equivalent to measuring $Z_1,Z_2$ and selecting $(+1,+1),(-1,-1)$. In the first case, only strings from $\ket*{D_1^6}$ can contribute, while in the second case only from $\ket*{D_5^6}$.
The respective post-measurement (non-normalized) branches are
\begin{align}
(\bra{00} \otimes I^{\otimes 4})\ket{A_6} = \frac{1}{\sqrt{12}}(\ket{1000}+\ket{0100}+\ket{0010}+\ket{0001}),
\end{align}
and
\begin{align}
(\bra{11} \otimes I^{\otimes 4})\ket{A_6} = -\frac{1}{\sqrt{12}}(\ket{0111}+\ket{1011}+\ket{1101}+\ket{1110}).
\end{align}
These two occur with total probability $\Pr(00)+\Pr(11)= 1/3 + 1/3=2/3$. The outcomes $01,10$ are discarded. At this step, we call the respective normalized states $\ket*{G_4},\ket*{\overline{G}_4}$. Finally, we measure the first qubit in the $\ket{\pm}$ basis. This corresponds to a $X_3$ measurement in the original ordering. Both outcomes are retained. We get the branches
\begin{align}
\begin{aligned}
(\bra{+}\otimes I^{\otimes 3})\ket*{G_4} &= \frac{1}{2\sqrt{2}}(\ket{000}+\ket{100}+\ket{010}+\ket{001}) = \frac{1}{\sqrt{2}} \ket{\chi} \\
(\bra{-}\otimes I^{\otimes 3})\ket*{G_4} &= \frac{1}{2\sqrt{2}}(-\ket{000}+\ket{100}+\ket{010}+\ket{001})= -\frac{1}{\sqrt{2}}Z^{\otimes 3}\ket{\chi}\\
(\bra{+}\otimes I^{\otimes 3})\ket*{\overline{G}_4} &= \frac{1}{2\sqrt{2}}(\ket{111}+\ket{011}+\ket{101}+\ket{110})= \frac{1}{\sqrt{2}}X^{\otimes 3}\ket{\chi}\\
(\bra{-}\otimes I^{\otimes 3})\ket*{\overline{G}_4} &= \frac{1}{2\sqrt{2}}(\ket{111}-\ket{011}-\ket{101}-\ket{110})= -\frac{1}{\sqrt{2}}Z^{\otimes 3}X^{\otimes 3}\ket{\chi}.
\end{aligned}
\end{align}
Hence the respective output states are (up to global phases)
\begin{align}
\ket{\phi_{00,+}} = \ket{\chi}, \quad \ket{\phi_{00,-}} = Z^{\otimes 3} \ket{\chi}, \quad
\ket{\phi_{11,+}} = X^{\otimes 3} \ket{\chi}, \quad
\ket{\phi_{11,-}} = Y^{\otimes 3} \ket{\chi}.
\end{align}
The protocol therefore concludes with the application of the conditioned Pauli correction among $I^{\otimes 3},X^{\otimes 3},Y^{\otimes 3},Z^{\otimes 3}$, and with the unitary Clifford conversion $\ket{\chi} \to \ket{\CCZ}$ (see \cref{eq:cleq}). The total success probability is
\begin{align}
\Pr_{\Lambda_6}(\psi^{\otimes 6} \to \CCZ) = \frac{1}{3} \Mlin(\psi).
\end{align}
This concludes the proof of achievability.
\end{proof}
\end{theorem}

\begin{remark} The symmetry imposed by the first stabilizer gives directly Pauli equivalence between all the different outcomes of $Z_1,Z_2,X_3$. Indeed, we have
\begin{align}
X^{\otimes 6} \ket{A_6} = - \ket{A_6}, \quad Z^{\otimes 6} \ket{A_6} = - \ket{A_6}.
\label{eq:inv6}
\end{align}
Now label with $t = 0$ the outcome $00$ of $Z_1,Z_2$ and with $t=1$ the outcome $11$. We also label the outcome of $X_3$ by $c\in\{0,1\}$. The measurement branches are, due to \cref{eq:inv6}
\begin{align}
\begin{aligned}
(\bra{t,t}\bra{x_c} \otimes I^{\otimes 3})\ket{A_6} &= (-1)^{c+t} (\bra{t,t}\bra{x_c} \otimes I^{\otimes 3})(X^{\otimes 6})^t (Z^{\otimes 6})^c \ket{A_6} \\
&= (-1)^{c+t} (\bra{0,0}\bra{x_c} \otimes I^{\otimes 3}) ( I^{\otimes 2} \otimes (X^{\otimes 4})^t (Z^{\otimes 4})^c) \ket{A_6} \\
&= (-1)^{c+t} (\bra{0,0}\bra{x_c} \otimes I^{\otimes 3}) ( I^{\otimes 2} \otimes (Z^{\otimes 4})^c(X^{\otimes 4})^t ) \ket{A_6} \\
&= (-1)^{c+t} (\bra{0,0}\bra{x_0} \otimes I^{\otimes 3}) (I^{\otimes 3} \otimes (Z^{\otimes 3})^c)( I^{\otimes 2} \otimes (X^{\otimes 4})^t ) \ket{A_6} \\
&= (-1)^{c+t} (\bra{0,0}\bra{x_0} \otimes (Z^{\otimes 3})^c (X^{\otimes 3})^t ) \ket{A_6},
\end{aligned}
\end{align}
where $\ket{x_0}\coloneqq\ket{+}$ and $\ket{x_1}\coloneqq\ket{-}$. Hence, we could have computed just one branch to prove the achievability, and then conclude by Clifford equivalence. This is the strategy we will pursue when $k=8$.
\end{remark}

\subsection{Eight-copy protocol}\label{app:eight}
We start by defining $q_X \coloneqq a^2 - b^2, q_Y \coloneqq a^2+b^2, q_Z \coloneqq 2ab$.
We first prove the following.
\begin{lemma}[structure of $\mathcal{M}_8$]\label{lem:orth8} The states
\begin{align}
\ket{R_0} \coloneqq \sqrt{\frac{7}{8}} \ket*{D_1^8} - \frac{1}{\sqrt{8}}\ket*{D_5^8},  \quad
\ket{R_1} \coloneqq \frac{\ket{D_2^8}-\ket*{D_6^8}}{\sqrt{2}}, \quad
\ket{R_2}\coloneqq \frac{1}{\sqrt{8}} \ket*{D_3^8} - \sqrt{\frac{7}{8}} \ket*{D_7^8},
\end{align}
form an orthonormal basis of $\mathcal{M}_8$.
Furthermore, the states
\begin{align}
\ket*{E_X} \coloneqq \frac{\ket{R_0}-\ket{R_2}}{\sqrt{2}}, \quad \ket{E_Y} \coloneqq \frac{\ket{R_0}+\ket{R_2}}{\sqrt{2}}, \quad \ket{E_Z} \coloneqq \ket{R_1},
\label{eq:es}
\end{align}
are also an orthonormal basis of $\mathcal{M}_8$. They are also Clifford equivalent, and are eigenstates of the Pauli measurement $\{X^{\otimes 8},Z^{\otimes 8}\}$ corresponding to the following syndrome table, restricted to the subspace $\mathcal{M}_8$
\begin{align}
\begin{array}{c|cc}
& X^{\otimes8} & Z^{\otimes8}\\
\hline
\ket{E_X} & +1 & -1\\
\ket{E_Y} & -1 & -1\\
\ket{E_Z} & -1 & +1
\end{array}.
\end{align}
The fourth syndrome eigenspace $(+1,+1)$ is orthogonal to $\mathcal{M}_8$.
\begin{proof}
By \cref{lem:dimmk}, we need three basis vectors for $\mathcal{M}_8$.
Therefore we have to show only that $R_j \in \mathcal{M}_8$ and $\braket{R_i}{R_j} = \delta_{ij}$. First, consider the qubit stabilizer states of the form
\begin{align}
\ket{s_\xi}  \coloneqq \frac{\ket{0} + \xi \ket{1}}{\sqrt{2}}, \quad \xi \in \{\pm1,\pm i\},
\end{align}
then, the overlap with the Dicke state of weight $w$ reads
\begin{align}
\braket*{s_\xi^{\otimes 8}}{D_w^8} = 2^{-4} \sqrt{\binom{8}{w}} (\xi^*)^w.
\end{align}
Substituting the definitions of $\ket{R_0},\ket{R_1},\ket{R_2}$, each overlap satisfies $\braket*{s_\xi^{\otimes 8}}{R_i} \propto 1 - (\xi^*)^4=0$. They are also orthogonal to the remaining stabilizers $\ket{0}^{\otimes 8},\ket{1}^{\otimes 8}$, since the states contain no vectors with weight $0$ or $8$. Since the Dicke states are already in the the symmetric subspace, this proves $R_j \in \mathcal{M}_8$. Orthogonality follows immediately since their supports are orthogonal. This proves the first part of the theorem.

To prove the second part, just recall that $X^{\otimes 8}\ket*{D_w^8} = \ket*{D_{8-w}^8}$ and $Z^{\otimes 8}\ket*{D_w^8} = (-1)^w\ket*{D_{w}^8}$, therefore $X^{\otimes 8}\ket{R_0} = - \ket{R_2},X^{\otimes 8} \ket{R_1} = - \ket{R_1}, X^{\otimes 8}\ket{R_2} = - \ket{R_0}$ and $Z^{\otimes 8}\ket{R_0}=-\ket{R_0}, Z^{\otimes 8}\ket{R_1}=\ket{R_1}, Z^{\otimes 8}\ket{R_2}=-\ket{R_2}$.
Substituting into \cref{eq:es}, this gives the syndrome table for $\ket{E_X},\ket{E_Y}, \ket{E_Z}$. Any state corresponding to the syndrome $(+1,+1)$ must then be orthogonal to 
$\mathcal{M}_8$. 

We now prove that $\ket{E_X},\ket{E_Y}, \ket{E_Z}$ are Clifford equivalent. We note that each is associated to a different syndrome sector, and they all belong to $\mathcal{M}_8$. In particular, $\ket{E_X}$ and $\ket{E_Z}$ differ by the syndrome mapping $X\to Z$. This can be implemented via $H^{\otimes 8}$ since $H^\dag XH=Z$. Since $\mathcal{M}_8$ is Clifford invariant (see \cref{eq:cliffinv}), and the syndrome identifies uniquely the three one-dimensional subspaces, we must have $H^{\otimes 8}\ket{E_X}\simeq \ket{E_Z}$ (direct calculation actually gives equality). On the other hand, $\ket{E_X}$ and $\ket{E_Y}$ are linked by the syndrome mapping $X \to Y$, which is implemented by $S^{\otimes 8}$ since $S^\dag X S = - Y$, and the sign does not matter since we take eight copies. Hence $S^{\otimes 8}\ket{E_X}\simeq \ket{E_Y}$ (direct calculation gives an imaginary phase $i$). This concludes the proof.
\end{proof}
\end{lemma}
We are now ready to derive the eight-copy stabilizer protocol.
\begin{theorem}[Eight-copy achievability]\label{thm:eightac} There exists an explicit eight-copy universal exact magic state concentration protocol $\Lambda_8$ with $\CCZ$ success probability
\begin{align}
\Pr_{\Lambda_8}(\psi^{\otimes 8} \to \CCZ) &= \frac{2}{3}\Mlin(\psi).
\end{align}
\begin{proof}
Expanding the eight copies of the input state in the Dicke basis we get
\begin{align}
\ket{\psi}^{\otimes 8} = \sum_{w=0}^8 \sqrt{\binom{8}{w}} a^{8-w} b^w \ket*{D_w^8}.
\label{eq:dec8}
\end{align}
The overlaps with the basis vectors of $\mathcal{M}_8$ introduced in \cref{lem:orth8} are
\begin{align}
\braket{R_0}{\psi}^{\otimes 8} = \sqrt{7} a^2 f_6, \quad \braket{R_1}{\psi}^{\otimes 8} = \sqrt{14} ab f_6, \quad \braket{R_2}{\psi}^{\otimes 8} = \sqrt{7} b^2 f_6,
\label{eq:overl}
\end{align}
where we recall that $f_6(a,b)= ab(a^4-b^4)$. Therefore, by \cref{lem:orth8}, since the $R_j$'s form an orthonormal basis of $\mathcal{M}_8$, we can write
\begin{align}
\Pi_{\mathcal{M}_8} \ket{\psi}^{\otimes 8} = \sqrt{7} f_6 \left( a^2 \ket{R_0} + \sqrt{2}ab\ket{R_1} + b^2 \ket{R_2} \right),
\end{align}
therefore
\begin{align}
\left\| \Pi_{\mathcal{M}_8} \ket{\psi}^{\otimes 8} \right\|^2 = 7|f_6|^2(|a|^4 + 2|a|^2|b|^2 + |b|^4) =7 |f_6|^2.
\label{eq:projval}
\end{align}
Now, we consider the eigenstates of $\{X^{\otimes 8},Z^{\otimes 8}\}$ inside the  $\mathcal{M}_8$ subspace as in \cref{lem:orth8}. Using \cref{eq:overl} we get the isotropic overlaps
\begin{align}
\braket{E_X}{\psi}^{\otimes 8} = \sqrt{\frac{7}{2}}f_6 q_X, \quad \braket{E_Y}{\psi}^{\otimes 8} = \sqrt{\frac{7}{2}}f_6 q_Y, \quad \braket{E_Z}{\psi}^{\otimes 8} = \sqrt{\frac{7}{2}} f_6 q_Z.
\label{eq:ovE}
\end{align}
In principle, this projection is not guaranteed to be a stabilizer operation, our protocol will nevertheless effectively project onto the basis states introduced in \cref{lem:orth8} $\{\ket{E_\mu}\}_{\mu=X,Y,Z}$ via an appropriate Pauli measurement.

We now measure $X^{\otimes 8}$ and $Z^{\otimes 8}$. We reject the outcome $(+1,+1)$, since by \cref{lem:orth8} it lies outside $\mathcal{M}_8$, and hence cannot support any universal exact magic-extraction branch. Note that, at this point, the post-measurement state does not lie necessarily within $\mathcal{M}_8$. Assume for now that we received outcome $(+1,-1)$. The resulting branch is
\begin{align}
\Pi_{+1,-1}\ket{\psi}^{\otimes 8}, \quad \Pi_{+1,-1} = \frac{1}{4}(I+X^{\otimes 8})(I-Z^{\otimes 8}).
\end{align}
From \cref{eq:dec8}, it is easy to see that setting $m_w \coloneqq a^{8-w} b^w$, the amplitude of the Dicke state $\ket*{D_w^8}$ in the post-measurement state becomes
\begin{align}
A_w = \begin{cases}
\frac{1}{2}(m_w+m_{8-w}), \quad &w \,\, \text{odd}, \\
0, \quad &w \,\, \text{even},
\end{cases}
\end{align}
so that
\begin{align}
\Pi_{+1,-1}\ket{\psi}^{\otimes 8} = \sum_{w=0}^8 \sqrt{\binom{8}{w}} A_w \ket*{D_w^8}.
\end{align}
Symmetry further gives $A_w = A_{8-w}$. Note that in principle, the resulting state still depends on $a,b$. We will show how to measure away this dependence with Paulis, while retaining some magic content. The idea is to perform a Pauli measurement on five qubits, and collapse onto a state that is Clifford equivalent to $\chi$, hence also to $\CCZ$.
In particular, we measure the first qubit in the $Z$ basis and qubits two-to-five in the $X$ basis, leaving three output qubits. By permutation symmetry, the precise choice of qubits is irrelevant.

For simplicity, we now identify the corresponding outcome with the binary notation. Let $t\in\{0,1\}$ be the
$Z$-measurement outcome, let $r\in\{0,1\}^4$ label the
$X$-measurement outcome, and set $j\coloneqq |r|$. We accept only
$j=1$ or $j=3$. The four-qubit $X$-measurement bra is
\begin{align}
\bra{r_X} = \frac14 \sum_{y\in\{0,1\}^4} (-1)^{r\cdot y}\bra y.
\end{align}
Therefore, the projected $(t,r)$-branch can be written as
\begin{align}
\sum_{k=0}^3 \sqrt{\binom{3}{k}} B_k^{(t,r)} \ket{D_k^3}, \quad B_k^{(t,r)} \coloneqq \frac{1}{4} \sum_{y\in\{0,1\}^4} (-1)^{r\cdot y} A_{t+|y|+k}
\label{eq:branch12}
\end{align}
Assume now $t=0$. Grouping $y$ by weight $s=|y|$ gives
\begin{align}
B_k^{(0,r)} = \frac{1}{4} \sum_{s=0}^4 \sum_{\substack{y\in\{0,1\}^4\\|y|=s}}
(-1)^{r\cdot y} A_{s+k} \equiv \frac{1}{4} \sum_{s=0}^4 K_s^{(4)}(r) A_{s+k},
\end{align}
where we defined the Krawtchouk coefficients \cite{MacWilliamsSloane1977}
\begin{align}
K_s^{(4)}(r) \coloneqq \sum_{\substack{y\in\{0,1\}^4\\|y|=s}}
(-1)^{r\cdot y}.
\end{align}
Since $K_s^{(4)}(r)$ depends only on $j=|r|$, we write it, with a small abuse of notation, as $K_s^{(4)}(j)$. We then have the following identity \cite{MacWilliamsSloane1977}
\begin{align}
\sum_{s=0}^4K_s^{(4)}(j)z^s = (1-z)^j(1+z)^{4-j},
\end{align}
For $j=1$, since $(1-z)(1+z)^3=1+2z-2z^3-z^4$, we have
\begin{align}
B_k^{(0,1)} = \frac{1}{4} \sum_{s=0}^4 K_s^{(4)}(1) A_{s+k} = \frac14
\left( A_k+2A_{k+1}-2A_{k+3}-A_{k+4} \right).
\end{align}
Using $A_{2\ell}=0$, $A_1=A_7$, and $A_3=A_5$, and defining
$\Delta\coloneqq A_1-A_3$ gives
\begin{align}
(B_0^{(0,1)},B_1^{(0,1)},B_2^{(0,1)},B_3^{(0,1)}) =\frac{\Delta}{4}(2,1,0,-1).
\end{align}
This completely specifies the brach for outcome $t=0,j=1$. Importantly, the state dependence is factorized, hence it enters only in the measurement probability, and not in the amplitudes of the state.
Note also that $\Delta = \frac{1}{2}(m_1+m_7-m_3-m_5) = \frac12 \left( a^7b+ab^7-a^5b^3-a^3b^5 \right) = \frac12ab(a^4-b^4)(a^2-b^2) = \frac12f_6q_X$. Substituting back the $B_k^{(0,1)}$'s into \cref{eq:branch12} gives the branch
\begin{align}
t=0,\,j=1 \implies \quad \frac{f_6q_X}{2\sqrt2}\ket{\tau_+}, \qquad \ket{\tau_+}\coloneqq\frac{1}{2\sqrt2}\left(2\ket{000} +\ket{001} +\ket{010} +\ket{100}-\ket{111}\right).
\end{align}
Now $\ket{\tau_+}=H^{\otimes3}\ket{\chi}$, therefore this branch is converted with a Clifford exactly to $\ket{\chi}$, and hence, due to \cref{eq:cleq} to $\ket{\CCZ}$ by another Clifford. Now, we argue that the other branches are also Clifford-equivalent to $\ket{\CCZ}$. In particular, we only considered the joint syndrome $(+1,-1,0,1)$. \cref{eq:ovE} implies that the resulting Kraus operator $L_{+1,-1,0,1}$ satisfies, on the symmetric subspace (since it holds for all $\psi^{\otimes 8}$)
\begin{align}
L_{+1,-1,0,1} \Pi_{\Sym^8(\mathbb{C}^2)} = \frac{1}{2\sqrt{7}} \ketbra{\tau_+}{E_X}.
\end{align}
This calculation shows that the subsequent Pauli measurements effectively annihilate the effect on any state orthogonal to $\ket{E_X}$, at least on the symmetric subspace.
Let's now keep $t,j$ fixed. We can change between the $(+1,-1),(-1,-1),(-1,+1)$ syndromes by applying a Clifford correction to put back onto $(+1,-1)$.
What changes however are the effective coefficients $a,b$ of the input state, which are also rotated by the same Clifford. 
More precisely, if the branch has zero probability, we can ignore it. If the branch has nonzero probability, then it must be supported within $\mathcal{M}_8$. Now, by \cref{lem:orth8}, we can change the $X^{\otimes 8},Z^{\otimes 8}$ syndrome of the input state with a transversal Clifford $C^{\otimes 8}$ to put it back on $(+1,-1)$. By \cref{lem:orth8}, there is only a one-dimensional subspace within $\mathcal{M}_8$ with this syndrome, hence we can restrict to the Kraus branch $L_{+1,-1,0,1}$.
Using the Clifford equivalence among $\ket{E_X},\ket{E_Y},\ket{E_Z}$ proved in \cref{lem:orth8}, this effectively gives the other two effects
\begin{align}
L_{-1,-1,0,1} \Pi_{\Sym^8(\mathbb{C}^2)} \simeq \frac{1}{2\sqrt{7}} \ketbra{\tau_+}{E_Y}, \quad L_{-1,+1,0,1} \Pi_{\Sym^8(\mathbb{C}^2)} \simeq \frac{1}{2\sqrt{7}} \ketbra{\tau_+}{E_Z}.
\end{align}
This clarifies why, for the sake of the derivation, we can assume that the initial syndrome is $(+1,-1)$. 

Let us now analyze the other values of $t,j$. We will repeat a similar reasoning as the one discussed in the remark below \cref{thm:sixprob}.
We want to show that the other cases are Clifford equivalent. 
Call the post-$(+1,-1)$ syndrome state $\ket{\theta}$. This satisfies $X^{\otimes 8}\ket{\theta} = \ket{\theta}$ and $Z^{\otimes 8}\ket{\theta} = -\ket{\theta}$. Now, measuring $t=1$ instead of $t=0$ corresponds to starting from $X_1 \ket{\theta}$. Due to the $X^{\otimes 8}$ symmetry, this is the same as applying $X_2...X_8$, but these just give a global phase $(-1)^j$ to the four-qubit $X$ measurement and the resulting state differs by $(-1)^j X^{\otimes 3}$. Consider now instead $j=3$. Using the same logic, $Z_2Z_3Z_4Z_5\ket{\theta}=-Z_1Z_6Z_7Z_8\ket{\theta}$, so the resulting state differs from $j=1$ only by $-(-1)^t Z^{\otimes 3}$. Putting all together, identifying the first syndrome with $\mu=X,Y,Z$ as in \cref{lem:orth8}, we derived, up to global phases
\begin{align}
L_{\mu,t,j} \Pi_{\Sym^8(\mathbb{C}^2)} \simeq \frac{1}{2\sqrt{7}}(X^{\otimes 3})^t (Z^{\otimes 3})^{(j-1)/2} \ketbra{\tau_+}{E_\mu}, \quad t\in \{0,1\} ,\, j \in \{1,3\}
\end{align}
or more suggestively
\begin{align}
L_{\mu,t,j} \Pi_{\Sym^8(\mathbb{C}^2)} = \frac{1}{2\sqrt{7}} C_{t,j} \ketbra{\CCZ}{E_\mu}, \quad C_{t,j} \in \mathcal{C}_3.
\end{align}
Therefore, the complete measurement effectively implements the projector onto $\mathcal{M}_8$:
\begin{align}
\Pi_{\Sym^8(\mathbb{C}^2)} \left(\sum_{\mu=X,Y,Z} \sum_{j=1,3} \sum_{t={0,1}} L_{\mu,t,j}^\dag L_{\mu,t,j} \right) \Pi_{\Sym^8(\mathbb{C}^2)}&= \frac{2 \cdot (4+4)}{28} \sum_{\mu=X,Y,Z} \ketbra{E_\mu} \\
&= \frac{4}{7} \Pi_{\mathcal{M}_8},
\end{align}
where we used that there are $2\bigl(\binom41+\binom43\bigr)=16$ total accepted values of $t,j$. Then, \cref{eq:projval} gives
\begin{align}
\Pr_{\Lambda_8}(\psi^{\otimes 8} \to \CCZ) = \frac{4}{7} \left\| \Pi_{\mathcal{M}_8} \ket{\psi}^{\otimes 8} \right\|^2 = \frac{2}{3} \Mlin(\psi).
\end{align}
To conclude, we also note that the outcomes $j=0,4$ are discarded since the resulting state still depends on the amplitude $A_k$ individually, hence it is not independent of the input. Outcome $j=2$ corresponds instead to a stabilizer state.
This concludes the proof.
\end{proof}
\end{theorem}

\section{Upper bounds up to $k=9$}
In this section we derive the tight upper bound on the $k=6$ success probability, and other representation-theoretic upper bounds when $6\leq k\leq9$. While we are only able to get tight achievability and optimality when $k=6$, the $k>6$ upper bounds show that the scaling in $\Mlin$ is effectively optimal up to $k=9$. When $k\geq 10$, other Clifford-invariant quantities can appear in the upper bound.

\subsection{Optimality when \texorpdfstring{$k=6$}{k=6}} 
In this section we show that the protocol we derived for $k=6$ achieves in fact optimal success probability among all $\CCZ$ exact magic state concentration protocols acting on $6$ input copies (see \cref{def:exact}).
We first prove a few preliminary lemmas. The following allows for a simplification for the evaluation of the properties of the state $\ket{A_6}$, since the respective Clifford-invariant measures can then be reduced to a smaller non-stabilizer state.

\begin{lemma}[Intermediate state $\eta$]\label{lem:intstate} We have $\ket{A_6} \simeq_{\mathrm{Cl}} \ket{\eta}\ket{-}\ket{1}$, where
\begin{align}
\ket{\eta} = \frac{\ket{0000}+\sum_{j=1}^4\ket{e_j}+\ket{1111}}{\sqrt{6}},
\end{align}
and $e_j \in \{0,1\}^4$ and has $1$ in position $j$ and it is zero otherwise.
\begin{proof}
Recall that
\begin{align}
\ket{A_6} = \frac{\ket*{D_1^6} - \ket*{D_5^6}}{\sqrt{2}} = \frac{1}{\sqrt{12}} \sum_{j=1}^6(\ket{e_j} - \ket{\mathbf{1}+e_j}).
\end{align}
Consider the linear change of variables $(x_1,...,x_6)\to(\ell_1,\ell_2,\ell_3,\ell_4,t,p)$ where
\begin{align}
\begin{aligned}
\ell_j &= x_{j+1} + x_6, \quad j=1,...,4, \\
t&=x_6, \\
p &= \sum_{j=1}^6 x_j.
\end{aligned}
\end{align}
This is implemetable via $\mathrm{CNOT}$ and $\mathrm{SWAP}$, hence it is a Clifford unitary. Now, the last coordinate $p=1$ is the same on both $e_j$ and $\mathbf{1}+e_j$ and gives the $\ket{1}$ on the last qubit. $5$ strings $e_j$ have $x_6 =0$, and one has $x_6=1$. On the first, the map simply implements a shift and gives $\sum_{j=1}^4\ket{e_j}$ and $\ket{0000}$, on the last, it gives $\ket{1111}$. Hence
\begin{align}
\sum_{j=1}^6\ket{e_j} \longrightarrow \left[\left( \ket{0000}+\sum_{j=1}^4\ket{e_j}\right)\ket{0} + \ket{1111}\ket{1}\right]\ket{1}.
\end{align}
The outcome is reversed for $\mathbf{1}+e_j$, giving
\begin{align}
\sum_{j=1}^6\ket{\mathbf{1}+e_j} \longrightarrow \left[\left( \ket{0000}+\sum_{j=1}^4\ket{e_j}\right)\ket{1} + \ket{1111}\ket{0}\right]\ket{1}.
\end{align}
Putting the two together
\begin{align}
\ket{A_6} \longrightarrow \frac{\ket{0000}+\sum_{j=1}^4\ket{e_j}-\ket{1111}}{\sqrt{6}} \ket{-}\ket{1}.
\end{align}
Applying a final controlled-$Z$ gate on any two of the first four qubits concludes the proof.
\end{proof}
\end{lemma}
The state $\eta$ has the following properties when we compute the expectation values of Paulis. Using the binary notation (see \cref{eq:binary}) we get
\begin{lemma}[Properties of $\eta$]\label{lem:propeta} For any $a\neq0$:
\begin{enumerate}[label=(\roman*)]
    \item $\langle W_a \rangle_\eta \in \left\{ 0, \pm \frac{1}{3} \right\}$,
    \item $\langle W_a \rangle_\eta \neq 0 \iff q(a) = 0$,
    \item $\nu(\eta)=4$, $P_2(\eta) =1/6$.
\end{enumerate}
\begin{proof}
Set $E \coloneqq \{0,e_1,e_2,e_3,e_4,\mathbf{1}\} \subseteq \mathbb{F}_2^4$, then $\ket{\eta} = \frac{1}{\sqrt{6}}\sum_{y\in E}\ket{y}$. Then for $x,z \in \mathbb{F}_2^4$
\begin{align}
\begin{aligned}
\langle W_a \rangle_\eta &= \frac{i^{x\cdot z}}{6} \sum_{y,y'\in E} (-1)^{z\cdot y} \braket{y}{y'+x} \\
&=\frac{i^{x\cdot z}}{6} \sum_{y \in E, y+x \in E} (-1)^{z\cdot y}.
\end{aligned}
\end{align}
Now if $y\in E$, then $y+x \in E$ iff $\exists y'\in E$ such that $x = y+y'$. But note that any vector in $\mathbb{F}_2^4$ can be written as the sum of two elements of $E$, and such pair decomposition is unique. Indeed summing the pairs gives exactly $16$ distinct elements. Now, if $x\neq 0$ this implies that the sum above runs over two elements, differing by $x$, hence, for some $y\in E$
\begin{align}
\langle W_a \rangle_\eta = \frac{i^{x\cdot z}}{6} ( (-1)^{z\cdot y}+(-1)^{z\cdot (y+x)}) = \begin{cases}
0, \quad &x\cdot z = 1, \\
\pm \frac{1}{3}, &x\cdot z = 0.
\end{cases}
\label{eq:w1}
\end{align}
If instead $x=0$, since we assumed $a\neq0$, we must have $z\neq 0$. Then, setting $t=|z| \in [4]$
\begin{align}
\begin{aligned}
\langle W_a \rangle_\eta &= \frac{1}{6} \sum_{y \in E} (-1)^{z\cdot y} \\
&= \frac{1+\sum_{j=1}^4 (-1)^{z_j} +(-1)^{t}}{6} \\
&= \frac{5- 2t +(-1)^{t}}{6} \in \left\{ \pm \frac{1}{3} \right\},
\end{aligned}
\label{eq:w2}
\end{align}
where we used that $\sum_{j=1}^4 (-1)^{z_j} = 4-2t$. \cref{eq:w1,eq:w2} prove (i) and (ii).
In particular, this implies that no nonidentity Pauli stabilizes $\eta$, since that would require $|\langle W_a \rangle_\eta|=1$, hence $\nu(\eta) =4$. To compute the purity $P_2$, we need to count how many Paulis have expectation value of magnitude $1/3$. Condition (ii) implies that it suffices to count how many $q(a)= x\cdot z =0$, with $x,z \in \mathbb{F}_2^4$. Since this is a single linear constraint, the solution space has dimension $3$, hence it contains $8$ vectors, therefore counting both $x,z$ we get $2^4 + (2^4-1)2^3=136$. Hence the magnitudes are $1 \times 1, 135 \times 1/3$ and the rest zero. Hence $P_2(\eta) = 2^{-4}\left( 1+\frac{135}{3^4} \right) = 1/6$. This proves also (iii).
\end{proof}
\end{lemma}

The state $\chi$, which is Clifford-equivalent to $\CCZ$ (recall \cref{eq:cleq}), has instead the following properties.

\begin{lemma}[Properties of $\chi$]\label{lem:propchi} We have:
\begin{align}
\nu(\chi) = 3, \quad P_2(\chi) = \frac{11}{32},
\end{align}
furthermore, all Paulis expectation values of $\chi$ satisfy $\langle W_a \rangle_\chi \in \{0,\pm \frac{1}{2},1\}$.
\begin{proof}
Write $\ket{\chi} = \frac{1}{2} \sum_{x\in T} \ket{x}$, where $T = \{0,e_1,e_2,e_3\} \subset \mathbb{F}_2^3$. Following the same reasoning as in \cref{lem:propeta} we get, for $x,z\in \mathbb{F}_2^3$ 
\begin{align}
\langle W_a \rangle_\chi = \frac{i^{x\cdot z}}{4} \sum_{y \in T, y+x \in T} (-1)^{z\cdot y} .
\end{align}
For $x\neq0$, the six pairwise nonzero differences are $e_1, e_2,e_3, e_1+e_2,e_1+e_3,e_2+e_3$. For each, we have a unique pair $x,x+y$, and the outcome is nonzero exactly when $x\cdot z = 0$, since there are $4$ such $z$, this gives other $6\cdot 4=24$ nonzero Pauli expectations with value $\pm \frac{1}{2}$, and $32$ with expectation zero.
If instead $x=0$ and $z\neq 0$, setting $t\equiv |z|$, proceeding as in \cref{lem:propeta} we get
\begin{align}
\begin{aligned}
\langle Z^z \rangle_\chi &= \frac{1+ \sum_{j=1}^3 (-1)^{z_j}}{4} \\
&= \frac{4-2t}{4} \\
&= 1 - \frac{t}{2} \in \left\{0,\pm\frac{1}{2} \right\},
\end{aligned}
\end{align}
thus no nonidentity Pauli stabilizes $\chi$, and $\nu(\chi)=3$. In particular, we have $4$ strings giving $\langle Z^z \rangle_\chi \in \{\pm \frac{1}{2}\}$ and $3$ giving zero. Hence the spectrum has $1$ once, $\pm \frac{1}{2}$, $28$ times, and $0$, $35$ times, therefore
\begin{align}
\sum_{a\in \mathbb{F}_2^6} |\langle W_a \rangle_\chi|^4 = 1 + 28 \left(\frac{1}{2}\right)^4 = \frac{11}{4},
\end{align}
which gives $P_2(\chi) = \frac{11}{32}$.
\end{proof}
\end{lemma}

The previous lemmas allow us to show that the optimal success probability of exact $\CCZ$ conversion when starting from the fixed state $A_6$ is upper bounded as follows.

\begin{lemma}[Upper bound on $A_6$] \label{lem:fixedup}The stabilizer conversion probability from $A_6$ to $\chi$ satisfies
\begin{align}
\Pr_{\mathrm{opt}}(A_6 \to \chi) \leq \frac{2}{3}.
\end{align}
\begin{proof}
By hypothesis of exact concentration towards $\chi$, we may assume that, up to a Clifford unitary, the final branch has the form $\ket{\chi}\ket{\zeta}$, for some pure state $\ket{\zeta}$, where discarding has been postponed to the end of the protocol. Purity of $\ket{\chi}$ guarantees this product form. Since stabilizer auxiliary qubits and Clifford unitaries do not change either the stabilizer nullity or the stabilizer purity, the only nontrivial changes along the branch can arise from Pauli measurements.

First, we show that in any successful branch there must be exactly one type-(iii) Pauli measurement, as defined in \cref{sec:stabform}. 

Assume by contradiction that there are only type-(i) and (ii) measurements, then on the final branch, under these restricted measurements, we must have the Clifford equivalence $\ket{A_6} \simeq_{\mathcal{C}} \ket{\chi}\ket{\zeta}$.
Due to \cref{lem:propeta,lem:intstate}, $\nu(A_6)=\nu(\eta)= 4$, and due to \cref{lem:propchi}, $\nu(\chi)=3$. Then, we would have $\nu(A_6) = \nu(\chi)+\nu(\zeta)$, hence $\nu(\zeta)=1$. Also, from \cref{lem:propeta,lem:propchi}, we must have $P_2(\eta)=P_2(A_6)=1/6$, $P_2(\chi)=11/32$. Hence multiplicativity of $P_2$ would give
\begin{align}
P_2(\zeta)= \frac{P_2(A_6)}{P_2(\chi)} =16/33.
\label{eq:pur}
\end{align}
But $\nu(\zeta)=1$ implies that $\ket{\zeta}$ is Clifford equivalent to a pure qubit tensored with a stabilizer state, and every pure qubit satisfies $P_2 \geq 2/3$, so, since the purity is one on stabilizer states, also $P_2(\zeta)\geq2/3$, which contradicts \cref{eq:pur}. Hence, we cannot have $\ket{A_6} \simeq_{\mathcal{C}} \ket{\chi}\ket{\zeta}$ and there must be at least one type-(iii) Pauli measurement.
Furthermore, there can be only one type-(iii) Pauli measurement, because $\nu(\chi)=3$, and each such measurement decreases the stabilizer nullity by at least one (see \cref{sec:stabform}). The rest of the proof is therefore about bounding the success probability of the type-(iii) measurement.

Now, by the same reasoning as above, before the type-(iii) measurement, the state must be Clifford equivalent to $\ket{\eta}\ket{s}$ for some stabilizer state $\ket{s}$. Since Cliffords map Paulis into Paulis, we can assume without loss of generality to have exactly this state before the type-(iii) measurement. Since the measurement is type-(iii), it commutes with all the stabilizers, thus it acts on $\ket{s}$ as $\pm1$. Thus, up to a global sign that swaps measurement outcomes, we can reduce the measurement directly to $\ket{\eta}$. This will not be relevant in the following, since we will care only about magnitudes of the Pauli expectation.

Let the measured Pauli be $P$ with outcome $s = \pm 1$. The corresponding projector is $\Pi_s = (I+sP)/2$, and the probability of outcome $s$ is $p_s = \bra{\eta}\Pi_s\ket{\eta} = (1 + s \langle P \rangle_\eta)/2$. For a Pauli $Q$ commuting with $P$, its expectation in the post-measurement state is
\begin{align}
\langle Q \rangle_s = \frac{\bra{\eta}\Pi_sQ\Pi_s\ket{\eta}}{\bra{\eta}\Pi_s\ket{\eta}} = \frac{\langle Q \rangle_\eta+ s\langle QP \rangle_\eta}{1 + s \langle P \rangle_\eta}.
\end{align}
If instead $\{Q,P\}=0$, then $\Pi_sQ\Pi_s=0$, so $\langle Q\rangle_s=0$.
We now have two cases:
\begin{itemize}
    \item $\langle P\rangle_\eta=0$: by property (ii) in \cref{lem:propeta}, we have $q(p)=1$, and both outcomes occur with probability $p_s=1/2$. For every $Q=W_u$ commuting with $P$ and $Q\neq I,P$ (so that the hypothesis in \cref{lem:propeta} holds), the relation $[u,p]=0$ gives $q(u+p)=q(u)+1$, so exactly one of $q(u)$ and $q(u+p)$ vanishes. Hence, by properties (i)–(ii) in \cref{lem:propeta}, exactly one of $\langle Q\rangle_\eta$ and $\langle QP\rangle_\eta$ is nonzero, with magnitude $1/3$, and therefore $|\langle Q\rangle_s|=1/3$. $Q=I,P$ give instead magnitude $1$.
    Thus every $Q$ Pauli expectation on the post-measurement state has magnitude in ${0,1/3,1}$, and in particular none has magnitude $1/2$. Since any state Clifford equivalent to $\ket{\chi}$ tensored with a stabilizer state has Pauli expectations of magnitude $1/2$ by \cref{lem:propchi}, neither outcome can correspond to a successful branch.
    \item $\langle P \rangle_\eta \neq 0$: by \cref{lem:propeta}, we can only have $\langle P \rangle_\eta = \varepsilon/3$ with $\varepsilon \in \{\pm 1\}$, therefore the two outcome probabilities are $p_{s=\varepsilon} = 2/3$ and $p_{s=-\varepsilon} = 1/3$. We only need to show that the outcome $s=-\varepsilon$ cannot succeed. 
    For every $Q = W_u$ and $Q\neq I,P$ commuting with $P$, by property (ii) in \cref{lem:propeta}, $q(u+p) = q(u)$, thus $\langle Q \rangle_\eta = 0 \iff \langle QP \rangle_\eta = 0$, and if they are nonzero, they both have magnitude $1/3$. For the outcome $-\varepsilon$ 
    \begin{align}
    \langle Q \rangle_{-\varepsilon} = \frac{\langle Q \rangle_\eta - \varepsilon \langle QP \rangle_\eta}{1 - \varepsilon \langle P \rangle_\eta} = \frac{\langle Q \rangle_\eta - \varepsilon \langle QP \rangle_\eta}{2/3} \in \left\{ 0,\pm 1 \right\}
    \end{align}
    If instead $Q=I,P$, then $|\langle Q \rangle_{-\varepsilon}|=1$, so the previous equation extends to all $Q$. Now, for every pure qubit state, computing the purity we must have $\sum_{Q\in \widehat{\mathcal{P}}_4} |\langle Q \rangle|^2=16$.
    Since in our case each term is either $\{0,1\}$, there must be at least $16$ Paulis with $|\langle Q \rangle|^2=1$. But then the stabilizer group of the output would have at least $4$ generators, hence $\nu =0$, so it is a stabilizer state and cannot be successful. The only possible outcome is then $\varepsilon$ with probability $2/3$.
\end{itemize}
This concludes the proof, since we showed that there is only one possible valid branch and that branch has probability $2/3$.
\end{proof}
\end{lemma}
Now, we reduce the general six-copy protocol to ones that effectively concentrate starting from $A_6$. Then, using the previous lemma we conclude.
\begin{theorem}[Optimality up to $k=6$]\label{thm:opt6} Among all universal exact $\CCZ$ magic-state concentration protocols acting on six copies of a pure qubit state, the optimal success probability must satisfy
\begin{align}
\Pr_{\mathrm{opt}}(\psi^{\otimes 6} \to \CCZ) \leq \frac{1}{3} \Mlin(\psi),
\end{align}
furthermore, the probability of extracting at least two $\CCZ$ when $k=6$, or extracting at least one $\CCZ$ when $k<6$ are both zero.
\begin{proof}
Consider an universal stabilizer protocol acting on $\psi^{\otimes 6}$ and producing an exact $\CCZ$ state with some nonzero success probability. Since $\ket{\CCZ} \simeq_{\mathcal{C}} \ket{\chi}$, it is enough to consider $\ket{\chi}$ as output. Assume the valid $\ket{\chi}$ branches are indexed by $r$, then, on outcome $r$, we must have, for some Kraus operator $K_r$
\begin{align}
K_r \ket{\psi}^{\otimes 6} = c_r(\psi)\ket{\chi}, \quad \forall \psi.
\end{align}
In particular, we have
\begin{align}
c_r(\psi) = \bra{\chi}K_r \ket{\psi}^{\otimes 6}, \quad \forall \psi.
\end{align}
Since the $\psi^{\otimes 6}$'s span the symmetric subspace, it is equivalent to consider instead the restricted Kraus operator $K_r|_{\Sym^k(\mathbb{C}^2)}$, since both equations above do not change. Now $c_r(\psi)$ implements exactly the isomorphism we introduced in \cref{eq:isom}. Furthermore $c_r(s)=0$ on all stabilizer states $s \in \mathrm{Stab}_1$. Therefore, with a slight abuse of notation, we have $c_r(a,b) \in \mathcal{I}(\mathrm{Stab}_1)_6$. Now, by \cref{lem:homg},      $\mathcal{I}(\mathrm{Stab}_1)_6 = f_6\mathbb{C}$. 
Hence $c_r(a,b)= \gamma_r f_6(a,b)$, for some constant $\gamma_r$. Therefore, on any stabilizer branch
\begin{align}
K_r \ket{\psi}^{\otimes 6} = \gamma_r f_6(a,b) \ket{\chi}.
\end{align}
When $k<6$ instead, due to \cref{lem:homg}, no non-trivial vanishing polynomial exists, hence there is no successful branch.
Now, using $\braket{A_6}{\psi}^{\otimes 6} = \sqrt{3}f_6(a,b)$, we get
\begin{align}
K_r \ket{\psi}^{\otimes 6} = \frac{\gamma_r}{\sqrt{3}} \braket{A_6}{\psi}^{\otimes 6} \ket{\chi},
\end{align}
hence, since it must hold for all $\psi$
\begin{align}
K_r\big|_{\Sym^6(\mathbb{C}^2)} = \frac{\gamma_r}{\sqrt{3}}  \ketbra{\chi}{A_6}.
\end{align}
Now, the protocol is fixed independently of the input state. So if we pick as input $\ket{A_6}$ (note that this operation is well defined, since the protocol must work on any $\psi^{\otimes 6}$, and therefore by linearity also on $A_6$), we get the following equivalent characterization of the $\gamma_r$'s for any stabilizer protocol $\Lambda$
\begin{align}
\Pr_{\Lambda}(A_6 \to \chi) = \sum_r \frac{|\gamma|^2_r}{3},
\end{align}
hence 
\begin{align}
\begin{aligned}
\Pr_{\Lambda}(\psi^{\otimes 6} \to \chi) &= |f_6(a,b)|^2 \sum_r |\gamma_r|^2 \\
&= 3 |f_6(a,b)|^2 \Pr_{\Lambda}(A_6 \to \chi) \\
&= \frac{1}{2} \Mlin(\psi) \Pr_{\Lambda}(A_6 \to \chi) \\
&\leq \frac{1}{3} \Mlin(\psi).
\end{aligned}
\end{align}
Where in the last line we used \cref{lem:fixedup}. Finally, stabilizer nullity (see \cref{lem:intstate,lem:propeta}) gives $4=\nu(A_6)\geq3m$ for any branch containing $m$ exact $\CCZ$ states, hence $m\leq1$.
This concludes the proof.
\end{proof}
\end{theorem}

\subsection{Optimal scaling in \texorpdfstring{$\Mlin$}{Mlin} up to \texorpdfstring{$k=9$}{k=9}}
In this section, we show that the scaling in $\Mlin$ is effectively optimal up to $k=9$ pure qubit-state input copies. Furthermore, it cannot be fundamentally improved for any target non-stabilizer state, thus not necessarily only for $\CCZ$.

\begin{proposition}[Optimal scaling when $6\leq k\leq 9$] \label{prop:opt6to9} Among all universal stabilizer protocols converting $\ket{\psi}^{\otimes k}$ exactly into any fixed non-stabilizer state $\ket{\tau}$, the optimal success probability satisfies
\begin{align}
\Pr_{\mathrm{opt}}(\psi^{\otimes k} \to \tau) \leq \frac{7(k-5)}{2(k+1)} \Mlin(\psi), \quad 6\leq k\leq9.
\end{align}
\begin{proof}
Recall that for any fixed stabilizer protocol, on any of its branches satisfies, in the Kraus representation (see \cref{eq:proj})
\begin{align}
K_r \Pi_{\Sym^k(\mathbb{C}^2)} = K_r \Pi_{\mathcal{M}_k}.
\end{align}
Therefore, for any stabilizer protocol $\Lambda$
\begin{align}
\begin{aligned}
\Pr_{\Lambda}(\psi^{\otimes k} \to \tau) &= \sum_r \left\| K_r \ket{\psi}^{\otimes k} \right\|^2 \\
&= \sum_r \left\| K_r \Pi_{\Sym^k(\mathbb{C}^2)} \ket{\psi}^{\otimes k} \right\|^2 \\
&= \sum_r \left\| K_r \Pi_{\mathcal{M}_k} \ket{\psi}^{\otimes k} \right\|^2 \\
&\leq \left\| \Pi_{\mathcal{M}_k} \ket{\psi}^{\otimes k} \right\|^2,
\label{eq:bound44}
\end{aligned}
\end{align}
where in the last line we used that $\sum_r K_r^\dag K_r \leq I$ when restricting to successful outcomes. In the rest of the proof we will therefore bound the last line. 
Now, define the Hermitian operator
\begin{align}
\Tilde{\Omega}_6 \coloneqq \frac{1}{4}(I - X^{\otimes 6}- Y^{\otimes 6}- Z^{\otimes 6}),
\end{align}
then
\begin{align}
\bra{\psi}^{\otimes 6} \Tilde{\Omega}_6  \ket{\psi}^{\otimes 6} = \frac{1}{2} \Mlin(\psi).
\label{eq:a6mlin}
\end{align}
Let us consider the operator
\begin{align}
A_k \coloneqq \Pi_{\Sym^k(\mathbb{C}^2)} \left( \Tilde{\Omega}_6 \otimes I^{\otimes(k-6)} \right) \Pi_{\Sym^k(\mathbb{C}^2)}.
\label{eq:ak}
\end{align}
Clearly, also $\bra{\psi}^{\otimes k} A_k  \ket{\psi}^{\otimes k} = \frac{1}{2} \Mlin(\psi)$. Now, for every $\ket{s}\in \mathrm{Stab}_1$ we have $\Omega_6 \ket{s}^{\otimes 6}=0$, hence also $A_6 \ket{s}^{\otimes 6}=0$, therefore
\begin{align}
A_k = \Pi_{\mathcal{M}_k} A_k \Pi_{\mathcal{M}_k}.
\end{align}
Now, due to \cref{eq:isom,lem:homg} (see also the discussion in \cref{lem:dimmk}), we have an isomorphism between different representations of the Clifford group
\begin{align}
\mathcal{M}_k \cong_{\mathcal{C}} \chi_6 \otimes \Sym^{k-6}(\mathbb{C}^2)^*,
\label{eq:cliffrepr}
\end{align}
with $\chi_6$ a scalar to be defined in the following derivation.
We already know that it is a linear isomorphism, to prove that it is a representation isomorphism, define
\begin{align}
B_k: \Sym^{k}(\mathbb{C}^2) \longrightarrow \mathbb{C}[a,b]_k, \quad \ket{\Phi} \mapsto p_{\Phi}(a,b).
\end{align}
Then, setting
\begin{align}
T_k \coloneqq B_k^{-1} \circ V_{f_6} \circ B_{k-6}, \quad V_{f_6}(q) \coloneqq f_6 q.
\end{align}
$T_k: \Sym^{k-6}(\mathbb{C}^2) \longrightarrow \mathcal{M}_k$ is a linear bijection, and it satisfies the intertwining relation
\begin{align}
C^{\otimes k} T_k = \chi_6(C) T_k C^{\otimes(k-6)},
\end{align}
where $\chi_6(C)$ is the scalar action $C f_6 = \chi_6(C) f_6$ on the one-dimensional irreducible $f_6$.
This proves \cref{eq:cliffrepr}.

Therefore, $\mathcal{M}_k$ is irreducible exactly when $\Sym^{k-6}(\mathbb{C}^2)$ is. Now set $m \coloneqq k-6$. We then consider $m=0,1,2,3$. Since the Clifford group is a 3-design \cite{Zhu_2017}, its commutant on $\Sym^{m}(\mathbb{C}^2)$ is equal to the corresponding $U(2)$ one, but $\Sym^{m}(\mathbb{C}^2)$ is an irreducible $U(2)$ representation, so the commutant contains only the identity. Now, in particular, the operator $A_k$ defined in \cref{eq:ak} belongs to the Clifford commutant, and has support only on $\mathcal{M}_k$. Therefore we must have, when $6 \leq k\leq 9$
\begin{align}
A_k = \lambda_k \Pi_{\mathcal{M}_k}, \quad \lambda_k \geq 0,
\label{eq:schur}
\end{align}
hence, due to \cref{eq:a6mlin}
\begin{align}
\frac{1}{2} \Mlin(\psi) = \bra{\psi}^{\otimes k} A_k  \ket{\psi}^{\otimes k} .
\label{eq:eq23}
\end{align}
Computing the Haar average on the LHS gives
\begin{align}
\frac{1}{2} \int d\psi \,\Mlin(\psi) = \frac{1}{4}\left( 1 - 3 \int d\psi \, x^{6}\right) = \frac{1}{4}\left( 1 - \frac{3}{7}\right) = \frac{1}{7}.
\end{align}
On the other hand, due to \cref{lem:dimmk}
\begin{align}
\int d\psi \, \bra{\psi}^{\otimes k} A_k  \ket{\psi}^{\otimes k} = \lambda_k \Tr\left[ \Pi_{\mathcal{M}_k} \frac{\Pi_{\Sym^k(\mathbb{C}^2)}}{k+1}\right] = \lambda_k\frac{\dim \mathcal{M}_k}{k+1} = \lambda_k\frac{k-5}{k+1},
\end{align}
then from \cref{eq:eq23} we conclude
\begin{align}
\lambda_k = \frac{k+1}{7(k-5)},
\end{align}
hence due to \cref{eq:eq23}
\begin{align}
\left\| \Pi_{\mathcal{M}_k} \ket{\psi}^{\otimes k} \right\|^2 &= \frac{1}{\lambda_k} \bra{\psi}^{\otimes k} A_k  \ket{\psi}^{\otimes k} \\
&=\frac{7(k-5)}{2(k+1)} \Mlin(\psi), \quad 6\leq k \leq 9,
\end{align}
which inserted into \cref{eq:bound44} concludes the proof of the upper bound.
\end{proof}
\end{proposition}

When $k\geq10$, the Clifford representation becomes reducible, hence additional Clifford invariants enter the bound. 
Note that the bounds for $k=6,8$ could have been obtained also through the direct basis expansion of the respective stabilizer orthogonal symmetric subspace.
Summarizing, we get the following bounds
\begin{align}
\begin{aligned}
\Pr_{\mathrm{opt}}(\psi^{\otimes 6} \to \tau) &\leq \frac{1}{2}\,\Mlin(\psi),\\
\Pr_{\mathrm{opt}}(\psi^{\otimes 7} \to \tau) &\leq \frac{7}{8}\,\Mlin(\psi),\\
\Pr_{\mathrm{opt}}(\psi^{\otimes 8} \to \tau) &\leq \frac{7}{6}\,\Mlin(\psi),\\
\Pr_{\mathrm{opt}}(\psi^{\otimes 9} \to \tau) &\leq \frac{7}{5}\,\Mlin(\psi).
\end{aligned}
\end{align}
Note that for $k=6$ the bound is not tight, and our target-dependent no-go gives the optimal one.

Moreover, the analysis carried out in the previous proof, allows us to prove something stronger than just tighter probability upper bounds. In particular, we can show that $\Mlin$ is the unique state-dependent quantity governing any Clifford-invariant protocol up to $k=9$ input copies.

\begin{proposition}[Clifford-invariant rigidity]\label{prop:rigid}
Let $6\leq k\leq9$, and consider a universal exact stabilizer protocol converting $\ket{\psi}^{\otimes k}$ into a fixed non-stabilizer state $\ket{\tau}$. If its total success effect
\begin{align}
E_k\coloneqq \sum_r K_r^\dag K_r
\end{align}
is invariant under the single-qubit Clifford group, then there exists a constant $c_k\geq0$, independent of $\psi$, such that
\begin{align}
\Pr(\psi^{\otimes k}\to\tau)
=
c_k\,\Mlin(\psi).
\end{align}
\end{proposition}

\begin{proof}
By universality, the restriction of $E_k$ to the symmetric subspace is supported on $\mathcal M_k$. Moreover, by \cref{eq:cliffrepr} and the irreducibility established above for $6\leq k\leq9$, Clifford invariance and Schur's lemma imply
\begin{align}
\Pi_{\Sym^k(\mathbb C^2)}E_k\Pi_{\Sym^k(\mathbb C^2)}
=
\alpha_k\Pi_{\mathcal M_k}
\end{align}
for some $\alpha_k\geq0$. Therefore,
\begin{align}
\Pr(\psi^{\otimes k}\to\tau)
&=
\alpha_k
\left|\Pi_{\mathcal M_k}\ket{\psi}^{\otimes k}\right|^2
\nonumber\\
&=
\alpha_k\frac{7(k-5)}{2(k+1)}\,\Mlin(\psi),
\end{align}
where in the last step we used the projector weight computed in the proof of \cref{prop:opt6to9}. The claim follows with
\begin{align}
c_k=\alpha_k\frac{7(k-5)}{2(k+1)}.
\end{align}
\end{proof}

\section{Asymptotic rate bound for universal many-to-one protocols}
In this section, we consider universal asymptotic magic-state distillation protocols built from independent many-to-one blocks, each acting on fresh copies of the input state and either producing one exact copy of a fixed non-stabilizer target $\tau$ or failing and being discarded. Denoting by $p_k^{\rm univ}(\psi\to\tau)$ the optimal success probability of a universal $k$-to-$1$ block, we define
\begin{align}
R_{\rm univ}^{(1)}(\psi\to\tau)
:=
\sup_{k\ge1}\frac{p_k^{\rm univ}(\psi\to\tau)}{k}.
\end{align}
This already gives the optimal asymptotic rate for arbitrary compositions of such blocks: if blocks of sizes $k_j$ and success probabilities $p_j$ are combined, then
\begin{align}
\frac{\sum_j p_j}{\sum_j k_j}
=
\sum_j\frac{k_j}{\sum_\ell k_\ell}\frac{p_j}{k_j}
\le
R_{\rm univ}^{(1)}(\psi\to\tau).
\end{align}

\begin{proposition}[Asymptotic optimality for universal many-to-one protocols]\label{thm:univone}
For every pure qubit state $\psi$ and every fixed non-stabilizer target state $\tau$,
\begin{align}
R_{\rm univ}^{(1)}(\psi\to\tau)
\le
\frac{1}{2}M_{\rm lin}(\psi).
\end{align}
In particular, for $\tau=\CCZ$,
\begin{align}
\frac{1}{12}M_{\rm lin}(\psi)
\le
R_{\rm univ}^{(1)}(\psi\to\CCZ)
\le
\frac{1}{2}M_{\rm lin}(\psi).
\end{align}
Thus, the eight-copy protocol is optimal up to a constant factor among all universal many-to-one protocols for exact $\CCZ$ distillation.
\end{proposition}

\begin{proof}
Consider any universal $k$-to-$1$ protocol for a fixed non-stabilizer target $\tau$, and let $E_k$ denote its success effect. Exact universality implies that, for every single-qubit stabilizer state $\phi$, the success probability on $\phi^{\otimes k}$ must vanish, since stabilizer operations cannot produce the non-stabilizer state $\tau$. Hence
\begin{align}
E_k|\phi\rangle^{\otimes k}=0.
\end{align}
Choosing $\phi$ to maximize its overlap with $\psi$ therefore gives
\begin{align}
p_k^{\rm univ}(\psi\to\tau)
\le
1-|\langle\phi|\psi\rangle|^{2k}.
\end{align}
Let $(x,y,z)$ be the Bloch vector of $\psi$ and set $t:=\max\{|x|,|y|,|z|\}$. Since the six pure qubit stabilizer states have Bloch vectors $\pm e_x,\pm e_y,\pm e_z$, their maximal fidelity with $\psi$ is $F_{\rm stab}(\psi)=(1+t)/2$. It follows that
\begin{align}
\frac{p_k^{\rm univ}(\psi\to\tau)}{k}
&\le
\frac{1-F_{\rm stab}(\psi)^k}{k}
\nonumber\\
&=
\bigl(1-F_{\rm stab}(\psi)\bigr)
\frac{1+F_{\rm stab}(\psi)+\cdots+F_{\rm stab}(\psi)^{k-1}}{k}
\nonumber\\
&\le
1-F_{\rm stab}(\psi)
=
\frac{1-t}{2}.
\end{align}
On the other hand, since $M_{\rm lin}(\psi)=\frac12(1-x^6-y^6-z^6)$ and $x^2+y^2+z^2=1$,
\begin{align}
x^6+y^6+z^6
\le
t^4(x^2+y^2+z^2)
=
t^4,
\end{align}
so that $M_{\rm lin}(\psi)\ge(1-t^4)/2$. Moreover, $t\ge1/\sqrt3$ and therefore
\begin{align}
\frac{1-t^4}{2}-(1-t)
=
\frac{1-t}{2}\bigl(t^3+t^2+t-1\bigr)
\ge0.
\end{align}
Thus $M_{\rm lin}(\psi)\ge1-t$, and hence
\begin{align}
\frac{p_k^{\rm univ}(\psi\to\tau)}{k}
\le
\frac12M_{\rm lin}(\psi).
\end{align}
Taking the supremum over $k$ proves the universal converse.
For $\tau=\CCZ$, the universal eight-copy protocol succeeds with probability $\frac23M_{\rm lin}(\psi)$, and therefore achieves the rate $M_{\rm lin}(\psi)/12$, proving the lower bound.
\end{proof}

\section{General asymptotic rate optimality up to logarithmic factors}\label{app:lowmagic}
In this section, we show that the block-repeated version of our protocol achieves an asymptotic rate that is optimal in scaling, up to logarithmic factors.
We first define the relative entropy of magic as \cite{Veitch_2014}
\begin{align}
D_{\mathbb{M}}(\rho) \coloneqq \min_{\sigma \in \mathrm{conv}(\mathrm{Stab}_n)} D(\rho\|\sigma),
\label{eq:defD}
\end{align}
where $D(\rho\|\sigma) \coloneqq \Tr[\rho(\log_2\rho-\log_2\sigma)]$, and
\begin{align}
\operatorname{conv}(\mathrm{Stab}_n)\coloneqq
\left\{\sum_i p_i\ket{s_i}\!\bra{s_i}:\ket{s_i}\in\mathrm{Stab}_n,
\,p_i\geq0,\,\sum_i p_i=1\right\},
\end{align}
is the stabilizer polytope, with $\mathrm{Stab}_n$ being all the possible $n$-qubit stabilizer states. We also define the regularized version as
\begin{align}
D^\infty_{\mathbb{M}}(\rho) \coloneqq \lim_{m\to\infty} \frac{1}{m} D_{\mathbb{M}}(\rho^{\otimes m}).
\end{align}
We use the following direct consequence of Ref.\,\cite{Rubboli_2024}:
\begin{align}
D^\infty_{\mathbb{M}}(\CCZ) = \log_2\left(\frac{16}{9}\right).
\label{dexact}
\end{align}
Indeed, Proposition 9 gives $D_{\mathbb{M}}(\mathrm{Toffoli})=\log_2(16/9)$, while Theorem 4 implies additivity on arbitrary tensor powers. Since $\mathrm{Toffoli}\simeq_{\mathcal{C}}\CCZ$, \cref{dexact} follows by Clifford invariance.
Via standard arguments, the regularized relative entropy of magic bounds asymptotic conversion rates (see \cref{def:asy}) as \cite{HORODECKI_2012,Rubboli_2024}
\begin{align}
R(\rho \to \sigma) \leq \frac{D^\infty_{\mathbb{M}}(\rho)}{D^\infty_{\mathbb{M}}(\sigma)}.
\end{align}
Therefore,
\begin{align}
R(\psi \to \CCZ)
\leq
\frac{D^\infty_{\mathbb{M}}(\psi)}{\log_2(16/9)}
\leq
\frac{D_{\mathbb{M}}(\psi)}{\log_2(16/9)},
\label{eq:targbound}
\end{align}
where the second inequality follows from subadditivity of $D_{\mathbb{M}}$. Note that \cref{def:asy} is stronger than the one considered in \cite{Rubboli_2024}, hence the upper bound applies (see also the remark at the end of \cref{sec:protdef}). 
We then have the following low-magic upper bound on the rate.

\begin{proposition}[Asymptotic low-magic upper bound]
When $\Mlin(\psi)\to0$,
\begin{align}
R(\psi \to \CCZ) \leq \frac{\Mlin(\psi)}{6\log_2(16/9)} \log_2\left(\frac{6e}{\Mlin(\psi)}\right)+O\left(\Mlin(\psi)^2\log_2\frac{1}{\Mlin(\psi)}\right),
\end{align}
or equivalently $R(\psi \to \CCZ) \leq \Tilde{O}(\Mlin(\psi))$, where $\Tilde{O}(x)\coloneqq O(x\log(1/x))$.
\begin{proof}
Let $\psi$ have Bloch vector $(x,y,z)$ and set $s\coloneqq\max\{|x|,|y|,|z|\}$. By Clifford invariance, we can assume $z=s\geq0$. Consider
\begin{align}
\sigma \coloneqq \frac{1+s}{2}\ketbra{0} + \frac{1-s}{2}\ketbra{1}.
\end{align}
Since $\sigma\in\mathrm{conv}(\mathrm{Stab}_1)$, by \cref{eq:defD},
\begin{align}
D_{\mathbb{M}}(\psi) \leq D(\psi\|\sigma) = h_2\left(\frac{1-s}{2}\right),
\label{eq:bound4}
\end{align}
where $h_2(q)\coloneqq -q\log_2 q - (1-q)\log_2(1-q)$.
Set $M\equiv\Mlin(\psi)=\frac12(1-x^6-y^6-z^6)$. Since $x^6+y^6+z^6\leq s^4(x^2+y^2+z^2)=s^4$, we have $M\geq\frac12(1-s^4)$, so $M\to0$ implies $s\to1$. Write $s=1-\eta$. Then $x^2+y^2=1-z^2=2\eta-\eta^2$, and hence $x^6+y^6\leq(x^2+y^2)^3=O(\eta^3)$. Therefore, since $(1-\eta)^6 = 1-6\eta + O(\eta^2)$
\begin{align}
M &= \frac12\left(1-x^6-y^6-(1-\eta)^6\right) \\
&= 3\eta+O(\eta^2),
\end{align}
so $\eta=M/3+O(M^2)$. Thus, since $\eta = 1-s$, using the standard expansion of the binary entropy
\begin{align}
h_2(q) = q\log_2\left(\frac{e}{q}\right)+O(q^2),
\end{align}
we obtain
\begin{align}
\begin{aligned}
h_2\left(\frac{1-s}{2}\right) = \frac{M}{6}\log_2\left(\frac{6e}{M}\right) + O\left( M^2\log_2\frac{1}{M} \right).
\end{aligned}
\end{align}
Combining this with \cref{eq:targbound,eq:bound4} gives
\begin{align}
R(\psi\to\CCZ) \leq \frac{M}{6\log_2(16/9)} \log_2\left(\frac{6e}{M}\right) + O\left( M^2\log_2\frac{1}{M} \right).
\end{align}
\end{proof}
\end{proposition}

To leading order, we get the low-magic upper bound $\approx 0.201 M \log_2(1/M)$, while our asymptotic rate obtained iterating the eight-copy protocol is (see \cref{cor2:rates}) $R \geq M/12 \approx 0.0833 M$. We also give a more explicit bound in the following.

\begin{proposition}[Explicit asymptotic upper bound]\label{prop:globalbound}
For any single-qubit pure state $\psi$ with $\Mlin(\psi) >0$,
\begin{align}
R(\psi\to\CCZ) \leq C\,\Mlin(\psi)\log_2\!\left(\frac{1}{\Mlin(\psi)}\right),
\end{align}
where
\begin{align}
C &\coloneqq
\frac{3(3-\sqrt{3})}{8\log_2(16/9)}\left[1+\frac{ \log_2\!\left(\frac{8e}{3(3-\sqrt{3})}\right)}{\log_2(9/4)}\right]
<1.81.
\end{align}

\begin{proof} Using the same definitions as in the previous proof, we have $s\geq 1/\sqrt{3}$. Moreover, $\Mlin(\psi)\geq\frac{1-s^4}{2}$. Then
\begin{align}
\frac{q}{\Mlin(\psi)} &\leq \frac{1-s}{1-s^4} = \frac{1}{(1+s)(1+s^2)} \\
&\leq \frac{1}{(1+1/\sqrt{3})(1+1/3)} = \frac{3(3-\sqrt{3})}{8}.
\end{align}
Thus, setting for simplicity $\gamma \coloneqq\frac{3(3-\sqrt{3})}{8}$, we have $q\leq\gamma\Mlin(\psi)$. Now using that
\begin{align}
h_2(q)\leq q\log_2\!\left(\frac{e}{q}\right),
\end{align}
and the fact that $q\mapsto q\log_2(e/q)$ is increasing for $0\leq q\leq1$, we obtain due to \cref{eq:bound4}
\begin{align}
D_{\mathbb M}(\psi) \leq \gamma\Mlin(\psi) \log_2\!\left(\frac{e}{\gamma\Mlin(\psi)}\right).
\end{align}
\cref{eq:targbound} then gives
\begin{align}
R(\psi\to\CCZ) &\leq \frac{\gamma}{\log_2(16/9)} \Mlin(\psi)
\log_2\!\left(\frac{e}{\gamma\Mlin(\psi)}\right).
\label{eq:glob}
\end{align}
Finally, for a single qubit $\Mlin(\psi)\leq4/9$, therefore $\log_2\!\left(1/\Mlin(\psi)\right)
\geq\log_2(9/4)$. Consequently,
\begin{align}
\log_2\!\left(\frac{e}{\gamma\Mlin(\psi)}\right)
&= \log_2\!\left(\frac{1}{\Mlin(\psi)}\right) +\log_2\!\left(\frac{e}{\gamma}\right) \\
&\leq \left[ 1+ \frac{\log_2(e/\gamma)}{\log_2(9/4)} \right]
\log_2\!\left(\frac{1}{\Mlin(\psi)}\right).
\end{align}
Substituting this into \eqref{eq:glob} proves the claim, with
\begin{align}
C &= \frac{3(3-\sqrt{3})}{8\log_2(16/9)} \left[1+\frac{\log_2\!\left(\frac{8e}{3(3-\sqrt{3})}\right)}{\log_2(9/4)}\right] \approx 1.8043<1.81.
\end{align}
\end{proof}
\end{proposition}

\end{document}